\documentclass[11pt]{article}
\usepackage[margin=1in]{geometry}
\usepackage{setspace}
\usepackage{graphicx}
\usepackage{svg} 
\graphicspath{{Fig/}}
\usepackage{subcaption}

\usepackage{amsmath}
\usepackage{mathtools} % after amsmath
\usepackage{amsthm}    % after mathtools
\usepackage{amsfonts, amssymb, amsbsy, mathrsfs}
\usepackage{bm}  %or just use \boldsymbol{}
\usepackage[makeroom]{cancel}
\usepackage{physics}

\usepackage[numbers,compress]{natbib} 
\usepackage{hyperref}
\usepackage{cleveref} % must after hyperref
\usepackage{autonum} % must after cleverref

\usepackage{xcolor}

\usepackage{tcolorbox}
\usepackage{dcolumn}
\usepackage{verbatim}
\usepackage{enumerate,enumitem}

\theoremstyle{plain}
\newtheorem{theorem}{Theorem}
\newtheorem{lemma}{Lemma}
\newtheorem{proposition}[lemma]{Proposition}
\newtheorem{definition}{Definition}

\theoremstyle{remark}

\usepackage{authblk,abstract} 
\newcommand{\defeq}{\stackrel{\text{def}}{=}}

\DeclareMathOperator{\supp}{supp}

\newcommand{\cc}{\mathscr{C}}
\newcommand{\calC}{\mathcal{C}}
\newcommand{\calA}{\mathcal{A}}

\newcommand{\calL}{\mathcal{L}}
\newcommand{\calT}{\mathcal{T}}

\usepackage[commandnameprefix=ifneeded]{changes}  
\setauthormarkup{}
\usepackage{xparse}

\NewDocumentCommand{\ZL}{o m}{%
	\IfNoValueTF{#1}%
	{\todo[color=blue!40, inline]{\textbf{Zhi:} #2}}% set default to inline 
	{\todo[color=blue!40, fancyline]{\textbf{Zhi:} #2}}% % default
}
\definechangesauthor[name={ZL}, color=blue]{ZL}

\begin{document}
	\title{\textbf{Explicit States with Two-sided Long-Range Magic}}
	\author[1]{Zhi Li}
	\affil[1]{IBM Research}
	
	\date{\vspace{-5ex}}  % no date
	\maketitle

	\begin{abstract}
		
		Nonstabilizerness, or magic, is a necessary resource for quantum advantage beyond the classically simulatable Clifford framework. Recent works have begun to chart the structure of magic in many-body states, introducing the concepts of long-range magic---nonstabilizerness that cannot be removed by finite-depth local unitary (FDU) circuits---and the magic hierarchy, which classifies quantum circuits by alternating layers of Clifford and FDUs. In this work, we construct explicit states that provably possess two-sided long-range magic, a stronger form of magic meaning that they cannot be prepared by a Clifford circuit and a FDU in either order, thus placing them provably outside the first level of the magic hierarchy. Our examples include the ``magical cat" state, $|\psi\rangle \propto |0^n\rangle + |+^n\rangle$, and ground states of certain nonabelian topological orders. These results provide new examples and proof techniques for circuit complexity, and in doing so, reveal the connection between long-range magic, the structure of many-body phases, and the principles of quantum error correction.
	\end{abstract}
	
	% \vspace{1em} 
	
	\section{Introduction}

	Concepts from quantum information have profoundly shaped our understanding of quantum many-body physics. A prime example is long-range entanglement, the cornerstone for defining topological phases of matter \cite{chen2010local}. Yet, entanglement is insufficient for quantum computational advantage: stabilizer states can be highly entangled yet classically simulatable \cite{gottesman1998heisenberg}. The additional resource enabling universal quantum computation is nonstabilizerness, or magic \cite{bravyi2005universal,veitch2014resource,howard2017application}. This leads to a fundamental question by analogy: can magic itself be organized into a long-range, topologically robust structure?

	Recent works have formalized this analogy, introducing long-range magic (or long-range nonstabilizerness) as nonstabilizerness that cannot be removed by finite-depth local unitary circuits (FDU) \cite{ellison2021symmetry,korbany2025long,wei2025long} (see also \cite{white2021conformal,andreadakis2026exact}).
	This property is naturally defined by the impossibility of state preparation using the circuit class $\mathsf{FDU}\circ\mathsf{Clifford}$ (first Clifford, then FDU). 
	This motivates considering the class $\mathsf{Clifford}\circ\mathsf{FDU}$, which is equally natural from the state preparation perspective.
	In fact, if one seeks a stronger definition of long-range magic such that Clifford postprocessing do not promote short-range magic to long-range, as Clifford is magic-free, then the hardness against $\mathsf{Clifford}\circ\mathsf{FDU}$ circuits becomes another minimal requirement.
	
	These classes are elegantly unified and generalized in the recently proposed magic hierarchy \cite{parham2025quantum}, which classifies circuits via alternating Clifford and FDU layers. 
	The hierarchy’s first level contains both $\mathsf{FDU}\circ\mathsf{Clifford}$ and $\mathsf{Clifford}\circ\mathsf{FDU}$ (termed $\mathsf{A_1CQ}$ and $\mathsf{A_1QC}$ in Ref.~\cite{parham2025quantum}), and higher levels contain further alternations.
	While the magic hierarchy provides a compelling scheme, proving actual lower bounds has been challenging.
	To date, known lower bounds have been limited to the $\mathsf{FDU}\circ\mathsf{Clifford}$ class \cite{korbany2025long,wei2025long,parham2025quantum}.
	Even the other class at the first level, $\mathsf{Clifford}\circ\mathsf{FDU}$, remains largely unexplored and calls for different techniques and observations.

	This raises a natural question: which states are outside the entire first level, that is, cannot be prepared by either $\mathsf{FDU}\circ\mathsf{Clifford}$ or $\mathsf{Clifford}\circ\mathsf{FDU}$\footnote{In this work, for the FDU, we focus on ``all-to-all local" rather than ``geometrically local", since (1) the Clifford would break the geometric structure anyway (2) the proof is usually simpler with geometrical locality.}?
	We call this stronger property \emph{two-sided long-range magic}. 
	In this work, we answer this question by providing explicit constructions for such states.
	
	\begin{enumerate}
		\item \textbf{The ``magical cat" state}. 
		We introduce $|\psi\rangle \propto \ket{0^n}+ \ket{+^n}$, a cat-like state whose branches are product states in conjugate Pauli bases. 
		Despite its simplicity, it cannot be prepared by either $\mathsf{Clifford} \circ \mathsf{FDU}$ or
		$\mathsf{FDU} \circ \mathsf{Clifford}$, even allowing constant approximation error.
		
		\item \textbf{Nonabelian topological orders}. For quantum double and string-net
		models on closed surfaces of genus $g \geq 1$ that admit no locality-preserving logical gates (e.g., transversal gates), we prove that no ground state can be prepared by $\mathsf{Clifford} \circ \mathsf{FDU}$ within inverse-polynomial error. 
		Combined with the $\mathsf{FDU} \circ \mathsf{Clifford}$ direction established in prior works~\cite{wei2025long, parham2025quantum}, we thus establish that such nonabelian topological orders lie outside the first level of the magic hierarchy.
		We also provide explicit examples satisfying the no-locality-preserving-logical-gates condition.
	\end{enumerate}

	Our constructions reveal a compelling parallel with the landscape of long-range entanglement, which features two paradigmatic classes: cat-like states (e.g., the GHZ state), whose entanglement stems from macroscopic superposition, and topologically ordered states (e.g., the toric code), whose entanglement arises from anyonic structure \cite{haah2016invariant,aharonov2018quantum,li2025much} and error correction \cite{bravyi2006lieb,bravyi2025much}. 
	Our examples of two-sided long-range magic mirror this dichotomy: the ``magical cat" state serves as a magic analogue of the GHZ state, while our second class comprises ground states of topological orders, but necessarily nonabelian. 
	Notably, the canonical examples of long-range entanglement---the GHZ state and toric code states---are themselves stabilizer states with no magic. Our constructions show that simple alterations of these templates already suffice for two-sided long-range magic.

	\section{ZX-cat State}
	Consider the following state:
	\begin{equation}\label{eq:psi}
		\ket{\psi}=\frac{1}{\sqrt{2\alpha}}(\ket{0^n}+\ket{+^n}),
	\end{equation}
	where $\alpha=1+2^{-\frac{n}{2}}$ is the normalization constant.
	This state has a cat-like structure, where its two components are product states defined in the computational ($Z$) basis and the Hadamard ($X$) basis, respectively. 
	We thus propose to name it the ZX-cat state. 
	Given its inherent long-range magic, we also colloquially refer to it as the ``magical cat" state\footnote{
		This state belongs to a class of generalized GHZ states in Ref.~\cite{dur2002effective}. It also shares similarity with the state $\ket{+^n}+\ket{\text{cluster}}$ considered in Ref.~\cite{zhang2024long}, to which our results also apply.
	}.
	
	\begin{theorem}
		The state $\ket\psi$ cannot be prepared by $\mathsf{Clifford}\circ\mathsf{FDU}$, or $\mathsf{FDU}\circ\mathsf{Clifford}$, within a constant approximation error. 
	\end{theorem}
	
	In the following, we prove the two no-go results in \cref{sec:CU,sec:UC} respectively. 
	The proofs are first given for exact preparation; the constant-error case is covered in Sec.~\ref{sec:ZXappro}.
	In \cref{sec:ZX2ndproof}, we provide an alternative proof for $\mathsf{Clifford}\circ\mathsf{FDU}$ via approximate error-correcting codes.
	
	\subsection{$\mathsf{Clifford}\circ\mathsf{FDU}$}\label{sec:CU}
	Suppose by contradiction that $\ket{\psi}=CU\ket{0^n}$, where $C$ is Clifford, $U$ has a finite depth $t=O(1)$.
	Consider $\ket{\phi} \defeq C^\dagger \ket{\psi}$.
	Then, on the one hand,
	\begin{equation}\label{eq:phiSRE}
		\ket{\phi} = U\ket{0^n};
	\end{equation}
	on the other hand,
	\begin{equation}\label{eq:phisumofstab}
		\ket{\phi}=\frac{1}{\sqrt{2\alpha}}(\ket{\phi_1}+\ket{\phi_2}),
	\end{equation}
	where $\braket{\phi_1}{\phi_2}=2^{-\frac{n}{2}}$, and $\ket{\phi_1}$ and $\ket{\phi_2}$ are stabilizer states with stabilizer groups:
	\begin{equation}
		G_1=\expval{C^\dagger Z_i C}, G_2=\expval{C^\dagger X_i C}.
	\end{equation}
	We will draw a contradiction from \cref{eq:phiSRE,eq:phisumofstab}.
	
	To build some intuition, we note that if $C=I$, then \cref{eq:phisumofstab} (which equals \cref{eq:psi} in this case) is a superposition of two ``super-selection sectors" and has long-range correlation, hence contradicting with \cref{eq:phiSRE}.
	For general $C$, the following lemma shows that $\ket{\phi_1}$ and $\ket{\phi_2}$ still belong to different ``super-selection sectors".
	\begin{lemma}\label{lem:twostab}
		Let $\ket{\eta_1}$ and $\ket{\eta_2}$ be two stabilizer states such that $\braket{\eta_2}{\eta_1}\neq 0$.
		
		(1) Let $P$ be a Pauli string, $\mel{\eta_2}{P}{\eta_1}\neq 0$, then $\abs{\braket{\eta_2}{\eta_1}}=\abs{\mel{\eta_2}{P}{\eta_1}}$;
		
		(2) For any operator $V$ such that $\norm{V}_\infty=1$, we have $\abs{\mel{\eta_1}{V}{\eta_2}} \leq 2^{a} \abs{\braket{\eta_1}{\eta_2}}$. Here $a=|\supp(V)|$.
	\end{lemma}
	\begin{proof}
		(1) It suffices to prove $|\mel{\eta_2}{P}{\eta_1}|\leq |\braket{\eta_2}{\eta_1}|$ (the inequality in the other direction follows from replacing $\ket{\eta_1}$ by $P\ket{\eta_1}$).
		In the following, we assume without loss of generality that $\ket{\eta_2}=\ket{0^n}$ and denote $\ket{\eta_1}$ by $\ket\eta$.
		We have
		\begin{equation}
			|\braket{0^n}{\eta}|^2=\braket{0^n}{\eta}\braket{\eta}{0^n}
			=\frac{1}{2^n}\sum_{g\in G} \expval{g}{0^n} 
			=\frac{1}{2^n}\sum_{g\in G_Z} \expval{g}{0^n}. 
		\end{equation}
		Here $G$ is the stabilizer group of $\ket{\eta}$ and $G_Z$ is the subgroup of $G$ that contains only tensor products of Pauli $Z$s (and the identity).
		In the last summation, each term must be $\pm 1$ a priori.
		However, in fact, each term must be +1;
		otherwise, we would have $g\ket{0^n}=-\ket{0^n}$, which implies $\braket{0^n}{\eta}=0$ (since $g\ket{\eta}=+\ket{\eta}$) and hence contradicts our assumption.
		Therefore,
		\begin{equation}
			|\braket{0^n}{\eta}|^2 = \frac{1}{2^n}|G_Z|.
		\end{equation}
		Now, noticing that conjugating a Pauli operator by another Pauli operator at most introduces a scalar factor, it follows that
		\begin{equation}
			|\mel{0^n}{P}{\eta}|^2
			=\braket{0^n}{P\eta}\braket{P\eta}{0^n}
			=\frac{1}{2^n}\sum_{g\in G_Z} \expval{PgP^\dagger}{0^n} 
			\leq \frac{1}{2^n}|G_Z|
			=|\braket{0^n}{\eta}|^2.
		\end{equation}
		
		(2) We decompose $V$ into a sum of Pauli strings: $V=\sum\lambda_i P_i$.
		There are at most $4^{a}$ terms in the expansion.
		$\norm{V}_\infty= 1$ implies $\sum |\lambda_i|^2 \leq 1$, since $2^n \sum |\lambda_i|^2 = \norm{V}_2^2\leq 2^n$.
		% Unitarity of .
		Hence, by (1) and Cauchy-Schwarz,
		\begin{equation}
			|\mel{\eta_1}{V}{\eta_2}| \leq \sum |\lambda_i\mel{\eta_1}{P_i}{\eta_2}| 
			\leq \sum |\lambda_i|\cdot |\braket{\eta_1}{\eta_2}|
			\leq 2^{a}|\braket{\eta_1}{\eta_2}|.  \tag*{\qedhere}
		\end{equation}
	\end{proof}
	
	\textbf{Remark}: The lemma also follows from the fact that, for two stabilizer states $\ket{\eta_i}$ with stabilizer groups corresponding to linear subspaces $S_i\subseteq \mathbb{F}_2^{2n}$, we have either $\braket{\eta_1}{\eta_2}=0$ or 
	\begin{equation}
		\abs{\braket{\eta_1}{\eta_2}}^2 = 2^{-n+\dim S_1\cap S_2}.
	\end{equation}
	
	As a result of \cref{lem:twostab}, a local ``stabilizer"\footnote{We use the quotation mark since $\ket{\phi}$ is not necessarily a stabilizer state in the ordinary sense, and the operator is not necessarily a Pauli string.} of $\ket{\phi}$ must also nearly stabilize $\ket{\phi_{1,2}}$. 
	\begin{lemma}\label{lem:componentwise}
		Let $\ket{\phi}=\frac{1}{\sqrt{2\alpha}}(\ket{\phi_1}+\ket{\phi_2})$ be as above.
		Let $V$ be Hermitian, $\norm{V}_\infty = 1$, $a=|\supp(V)|$.
		If $V\ket{\phi}=\ket{\phi}$, then
		\begin{equation}\label{eq:componentwise}
			\expval{V}{\phi_1},~\expval{V}{\phi_2}= 1 - O(2^{a-\frac{n}{2}}).
		\end{equation}
	\end{lemma}
	\begin{proof}
		We have:
		\begin{equation}
			1=\expval{V}{\phi}
			= \frac{1}{2\alpha}\left( \expval{V}{\phi_1}+\expval{V}{\phi_2} +2\Re\mel{\phi_1}{V}{\phi_2}    \right).
		\end{equation}
		Due to \cref{lem:twostab}, $|\mel{\phi_1}{V}{\phi_2}|\leq 2^{a-\frac{n}{2}}$, hence: 
		\begin{equation}
			\begin{aligned}
				1\geq {\expval{V}{\phi_1}} 
				={2\alpha-\expval{V}{\phi_2}-2\Re\mel{\phi_1}{V}{\phi_2}}
				\geq 1+ 2^{1-\frac{n}{2}}-2^{1+a-\frac{n}{2}},
			\end{aligned}
		\end{equation}
		which implies \cref{eq:componentwise} for $\ket{\phi_1}$.
		The claim for $\ket{\phi_2}$ follows symmetrically.
	\end{proof}
	
	\begin{figure}
		\centering
		\includegraphics[width=0.6\linewidth]{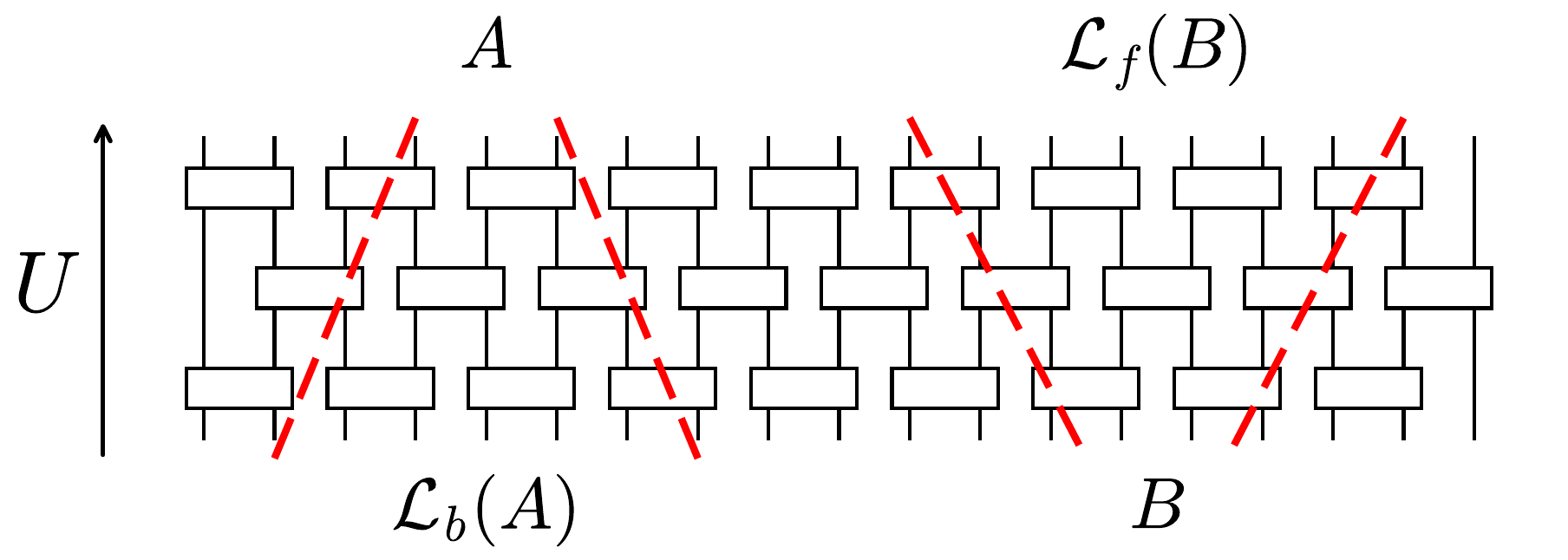}
		\caption{The backward light cone $\calL_b$ and the forward light cone $\calL_f$.}
		\label{fig:lightcone}
	\end{figure}
	
	Now consider the operator $V=UZU^\dagger$, where $Z$ is a single-qubit Pauli $Z$ operator.
	Due to the locality of the circuit $U$, $\supp(V)\subseteq \calL_f(\supp(Z))$, the forward light cone of $\supp(Z)$, see \cref{fig:lightcone}.
	Hence $a=|\supp(V)|\leq 2^t$.
	We also have $\expval{V}{\phi}=1$ due to \cref{eq:phiSRE}, 
	hence \cref{lem:componentwise} implies:
	\begin{equation}
		\expval{V}{\phi_1} = 1 - O(2^{a-\frac{n}{2}}).
	\end{equation}
	It follows that at least one stabilizer of $\ket{\phi_1}$ must be fully supported in $\supp(V)$: otherwise, $\ket{\phi_1}$ reducing to $\supp(V)$ must be maximally mixed, implying $\expval{V}{\phi_1} \propto \Tr(V) =0$. 
	We denote this stabilizer operator as $g$ ($g\neq I$), then:
	\begin{equation}\label{eq:expvalg}
		\expval{g}{\phi} 
		= \frac{1}{2\alpha}\left( \expval{g}{\phi_1}+\expval{g}{\phi_2} +2\Re\mel{\phi_1}{g}{\phi_2}    \right)
		= \frac{1}{2} + O(2^{a-\frac{n}{2}}).
	\end{equation}
	Here, we use $\expval{g}{\phi_2}=0$ (following from $G_1\cap G_2=\{I\}$) and \cref{lem:twostab} again.
	
	Repeating the above argument for another single-qubit Pauli operator $Z'$, we find another stabilizer $g'$ of $\ket{\phi_1}$, such that $\supp(g')\subseteq \calL_f(\supp(Z'))$ and
	\begin{equation}\label{eq:expvalgp}
		\expval{g'}{\phi} = \frac{1}{2}  + O(2^{a-\frac{n}{2}}).
	\end{equation}
	Moreover, we can choose $Z'$ so that the backward light cones $\calL_b(\supp(g))$ and $\calL_b(\supp(g'))$ are disjoint.
	(This can be done since $t=O(1)$.)
	Note that $gg'$ is also a stabilizer of $\ket{\phi_1}$, and $gg'\neq I$ due to the empty overlap of the supports, hence:
	\begin{equation}\label{eq:expvalggp}
		\expval{gg'}{\phi} = \frac{1}{2}  + O(2^{2a-\frac{n}{2}}).
	\end{equation}
	
	However, \cref{eq:expvalg,eq:expvalgp,eq:expvalggp} are incompatible.
	More precisely, due to our choice of two regions and the finite-depth assumption on $U$, there must be no correlation between $g$ and $g'$ in $\ket{\phi}$:
	\begin{equation}
		\expval{gg'}{\phi} = \expval{g}{\phi}\cdot\expval{g'}{\phi}= \frac{1}{4} + O(2^{a-\frac{n}{2}}).
	\end{equation}
	This is in contradiction with \cref{eq:expvalggp}.

	\subsection{$\mathsf{FDU}\circ\mathsf{Clifford}$}\label{sec:UC}

	We assume by contradiction that $|\psi\rangle = UC\ket{0^n}$ where $C$ is a Clifford and $U$ is a depth-$t$ circuit.
	We denote\footnote{Here we slightly abuse the notation; they are different from those defined in \cref{sec:CU}.} $\ket{\phi} = U^\dagger \ket{\psi}$, $\ket{\phi_1}= U^\dagger \ket{0^n}$, $\ket{\phi_2}= U^\dagger \ket{+}^{\otimes n}$, and $\ket{S}=C\ket{0^n}$.
	Therefore, on the one hand, $\ket{S}$ is a stabilizer state;
	on the other hand, 
	\begin{equation}
		\ket{S}= \ket{\phi} = \frac{1}{\sqrt{2\alpha}}(\ket{\phi_1}+ \ket{\phi_2}).
	\end{equation}

	As $\ket{\phi_{1,2}}$ are macroscopically different, the result of \cref{lem:componentwise} still applies, namely, local stabilizers of $\ket{S}$ also act almost trivially on $\ket{\phi_{1,2}}$.
	\begin{lemma}\label{lem:lifting}
		Let $g \in \mathrm{Stab}(\ket{S})$ be a Pauli stabilizer of $\ket{S}$, then 
		\begin{equation}\label{eq:lifting}
			\expval{g}{\phi_1},~\expval{g}{\phi_2}= 1 - O(2^{\frac{a-n}{2}}).
		\end{equation}
		Here, $a=|\calL_f(\mathrm{supp}(g))|$ where $\calL_f(\cdot)$ denotes the forward light cone, see \cref{fig:lightcone}. 
	\end{lemma}
	\begin{proof}
		Define the operator $V = UgU^\dagger$, then 
		$V\ket{\psi}=\ket{\psi}$ and $\supp(V)\subseteq \calL_f(\supp(g))$.
		Using the same argument for \cref{lem:componentwise}, we have $\expval{V}{0^n}=1 - O(2^{\frac{a-n}{2}})$.
		The same equation holds for $\expval{V}{+^n}$.
		They are exactly \cref{eq:lifting}.
	\end{proof}
	
	\begin{lemma}\label{lem:factor}
		Let $A, B$ be two sets of qubits such that their forward light cones are disjoint. Let $P_A, P_B$ be Hermitian operators on $A, B$ respectively, such that $\norm{P_A}_\infty=\norm{P_B}_\infty=1$. If $(P_A \otimes P_B)\ket{\phi_1} \approx \ket{\phi_1}$ up to a small enough $\epsilon$ error, then $P_A\ket{\phi_1} \approx \pm \ket{\phi_1}$ and $P_B\ket{\phi_1} \approx \pm \ket{\phi_1}$, both up to $\epsilon$.
		Same conclusion holds for $\ket{\phi_2}$.
	\end{lemma}
	\begin{proof}
		Let $Q_A = UP_A U^\dagger$ and $Q_B = UP_B U^\dagger$, which are supported on $\calL_f(A)$ and $\calL_f(B)$ respectively. 
		We have $(Q_A \otimes Q_B)\ket{0^n} \approx \ket{0^n}$ up to $\epsilon$. 
		Since $Q_A, Q_B$ act on disjoint qubits, we have
		\begin{equation}
			\expval{Q_A}{0^{n}} \expval{Q_B}{0^{n}} = \expval{Q_A\otimes Q_B}{0^n} \geq 1 -O(\epsilon^2).
		\end{equation}
		This implies either $\expval{Q_A}{0^n},\expval{Q_B}{0^n} \geq 1 -O(\epsilon^2)$, or $\expval{Q_A}{0^n},\expval{Q_B}{0^n} \leq -1 + O(\epsilon^2)$.
		In the former case, $Q_A\ket{0^n} \approx \ket{0^n}$ and $Q_B\ket{0^n} \approx \ket{0^n}$, both up to $\epsilon$.
		In the latter case, we get an extra minus sign.
		Acting by $U^\dagger$ returns the desired relations for $P_A, P_B$ and $\ket{\phi_1}$.
	\end{proof}

	We choose qubits $i$ and $j$ with backward light cones $B_i, B_j$ such that their forward light cones are disjoint: $\calL_f(B_i) \cap \calL_f(B_j) = \emptyset$.
	We can estimate the mutual information $I(B_i,B_j)_{\ket S}$ as follows:
	\begin{equation}\label{eq：MI}
		I(B_i,B_j)_{\ket S} 
		\geq
		I(i,j)_{U\ket S}
		= 
		I(i,j)_{\ket \psi} 
		=\frac{3}{4}\log_2 3+1-\frac{\sqrt{2}\mathrm{arctanh}\frac{2 \sqrt{2}}{3} }{2\log 2}  -e^{-\Omega(n)}
		>0. 
	\end{equation}
	Here, the first step is due to the data-processing inequality, and the third step is by explicit calculation.
	Since $\ket{S}$ is a stabilizer state, nonzero mutual information implies that there exists a Pauli stabilizer $g=P_i \otimes P_j$, supported on $B_i \cup B_j$, such that $P_i$ and $P_j$ are not stabilizers of $\ket{S}$.
	
	We will draw a contradiction as $n\to\infty$ for any fixed $t$, the depth of $U$.
	The sizes of the relevant light cones are bounded by $n$-independent constants, and we omit their dependence in the following for convenience (for example, $2^{a-n}$ will be $e^{-\Omega(n)})$).
	
	By Lemma~\ref{lem:lifting}, $g$ must stabilize both $\ket{\phi_1}$ and $\ket{\phi_2}$, up to $e^{-\Omega(n)}$:
	\begin{align}
		g\ket{\phi_1} \approx \ket{\phi_1},~~~~
		g\ket{\phi_2} \approx \ket{\phi_2}.
	\end{align}
	By Lemma~\ref{lem:factor}, we have:
	\begin{align}\label{eq:Pphi0}
		P_i\ket{\phi_1} \approx \ket{\phi_1},~~~
		P_j\ket{\phi_1} \approx \ket{\phi_1}.
	\end{align}
	Here, to fix the sign, we can replace $P_i$ and $P_j$ by $-P_i$ and $-P_j$ if necessary.
	Similarly,
	\begin{align}\label{eq:Pphi+}
		P_i\ket{\phi_2} \approx -\ket{\phi_2},~~~~
		P_j\ket{\phi_2} \approx -\ket{\phi_2}.
	\end{align}
	Note that we no longer have the freedom to change the signs in \cref{eq:Pphi+} since we might already have done that in \cref{eq:Pphi0}.
	In fact, the signs in \cref{eq:Pphi+} must be $-1$.
	Otherwise, $P_i$ would act as an identity on both $\ket{\phi_1}$ and $\ket{\phi_2}$, so $P_i$ must also stabilize their superposition $\ket{S}$. 
	This contradicts our choice of $g$, which requires that $P_i$ is not a stabilizer of $\ket{S}$.
	
	\Cref{eq:Pphi0,eq:Pphi+} imply that $\ket{\phi_1}$ and $\ket{\phi_2}$ are nearly orthogonal on $B_i$.
	More precisely, defining $\rho_{1,2}=\Tr_{B_i^c}(\ketbra{\phi_{1,2}})$, then 
	\begin{equation}
		F(\rho_1,\rho_2)=e^{-\Omega(n)}.
	\end{equation}
	On the other hand, recalling $\ket{\phi_1}= U^\dagger \ket{0^n}$, $\ket{\phi_2}= U^\dagger \ket{+}^{\otimes n}$, data processing inequality then implies that 
	\begin{equation}\label{eq:DPIF}
		F(\rho_1,\rho_2) \geq F(\ketbra{0^n}_{\calL_f(B_i)}, \ketbra{+^n}_{\calL_f(B_i)})=2^{-\frac{1}{2}|\calL_f(B_i)|}.
	\end{equation}
	This is a contradiction, as the r.h.s. is a $n$-independent constant.

	\subsection{Approximated Case}\label{sec:ZXappro}
	The proof for $\mathsf{Clifford}\circ\mathsf{FDU}$ in \cref{sec:CU} can easily handle the constant approximation error, so we focus on the $\mathsf{FDU}\circ\mathsf{Clifford}$ case.
	
	We follow \cref{sec:UC} closely and define $\ket{\phi}$, $\ket{\phi_{1,2}}$, and $\ket{S}=C\ket{0^n}$ same as above. $\ket{S}$ is still a stabilizer state. 
	The only difference is that now $\ket{S}\approx \ket{\phi}$, up to $\delta=O(1)$.
	
	We pick two qubits $i$ and $j$ similarly.
	Then, as a robust version of \cref{eq：MI}, we have:
	\begin{equation}
		I(B_i,B_j)_{\ket S} 
		\geq I(i,j)_{U\ket S}
		\approx I(i,j)_{\ket \psi} 
		>0. 
	\end{equation}
	Importantly, the dimension-dependent factor in the Fannes–Audenaert inequality does not matter in the second step, as we are only concerning qubits $i$ and $j$.
	Therefore, we still have the Pauli stabilizer $g=P_i \otimes P_j\in \mathrm{Stab}(\ket{S})$ such that $P_i,P_j\notin \mathrm{Stab}(\ket{S})$. 
	
	By a robust version of \cref{lem:lifting}, $g$ must approximately stabilize both $\ket{\phi_1}$ and $\ket{\phi_2}$:
	\begin{align}
		g\ket{\phi_1} \approx \ket{\phi_1},~~~~
		g\ket{\phi_2} \approx \ket{\phi_2}.
	\end{align}
	Here and in the following, the approximation errors are $O(\delta)+e^{-\Omega(n)}$, which is $O(\delta)$.
	By Lemma~\ref{lem:factor}, we have:
	\begin{equation}\label{eq:Pphiapprox}
		\begin{aligned}
			P_i\ket{\phi_1} \approx \ket{\phi_1},~~~~&P_j\ket{\phi_1} \approx \ket{\phi_1},\\
			P_i\ket{\phi_2}\approx -\ket{\phi_2},~~~~&P_j\ket{\phi_2}\approx - \ket{\phi_2}.
		\end{aligned}
	\end{equation}
	Here in the second line, we have similarly ruled out the possibility of $+1$ signs.
	
	Now let us pick $m$ qubits such that their backward light cones $\{B_i\}$ have mutually disjoint forward light cones $\{\calL_f(B_i)\}$.
	We can always achieve $m=\Theta(n)$ (more precisely , $m=\frac{n}{2^{4t}}$) via a greedy algorithm.
	We repeat the argument for each pair $(i, j)$ and get Pauli operators $P_i^{(j)}$ and $P_j^{(i)}$.
	A priori, $P_i$ also depends on the choice of $j$ and vice versa, so we have the superscripts.
	However, we claim that $P_i^{(j)}$ is the same Pauli (modulo stabilizers of $\ket{S}$) for all $j$.
	Indeed, for any $j\ne j'$, it follows from \cref{eq:Pphiapprox} that $P_i^{(j)}P_i^{(j')}$ approximately stabilizes both $\ket{\phi_1}$ and $\ket{\phi_2}$, hence approximately stabilizes $\ket{S}$. 
	Since $P_i^{(j)}P_i^{(j')}$ is Pauli, and $\ket{S}$ is a stabilizer state, approximate stabilization implies exact stabilization: $P_i^{(j)}P_i^{(j')}\in\mathrm{Stab}(\ket{S})$.
	In summary, we have a family of Paulis $\{P_i|i=1,\cdots, m\}$ such that $P_iP_j$ is a stabilizer for $\ket{S}$ for every $i$ and $j$.
	
	We define two stabilizer states $\ket{S_1}$ and $\ket{S_2}$ as:
	\begin{equation}
		\ket{S_1}=\frac{1+P_1}{\sqrt{2}}\ket{S},~~~~\ket{S_2}=\frac{1-P_1}{\sqrt{2}}\ket{S},
	\end{equation}
	hence $\ket{S}=\frac{1}{\sqrt{2}}(\ket{S_1}+\ket{S_2})$.
	By definition of $\ket{\phi_{1,2}}$ and \cref{eq:Pphiapprox}, we have:
	\begin{equation}\label{eq:Sphi}
		\ket{S_1} \approx \frac{1}{\sqrt{\alpha}}  \frac{1+P_1}{2}(\ket{\phi_1}+\ket{\phi_2})
		\approx \ket{\phi_1},
	\end{equation}
	and similarly $\ket{S_2}\approx \ket{\phi_2}$.
	In particular, the above defined $\ket{S_1}$ and $\ket{S_2}$ are nonzero.
	Importantly, since $P_iP_j=1$ on $\ket{S}$, it follows that $\ket{S_1}$ and $\ket{S_2}$ must be simultaneous eigenvectors of each $P_i$:
	\begin{equation}\label{eq:simul}
		\prod \frac{1+P_i}{2}\ket{S_1}=\ket{S_1},~~~~\prod \frac{1-P_i}{2}\ket{S_2}=\ket{S_2}.
	\end{equation}
	
	We define, for each $i$,
	\begin{equation}
		a_i=\expval{\frac{1+P_i}{2}}{\phi_1}.
	\end{equation}
	The disjointness of the light cones implies the following exact factorization:
	\begin{equation}\label{eq:fact}
		\expval{\prod \frac{1+P_i}{2}}{\phi_1} =\prod a_i.
	\end{equation}
	By \cref{eq:Sphi,eq:simul}, we have
	\begin{equation}
		\prod a_i \geq 1-O(\delta).
	\end{equation}
	Similarly, defining $b_i=\expval{\frac{1-P_i}{2}}{\phi_2}$, we get $\prod b_i \geq 1-O(\delta)$.
	By the AM-GM inequality, we have, for sufficiently small $\delta$,
	\begin{equation}\label{eq:totalerror}
		\sum a_i \geq m(\prod a_i)^{\frac{1}{m}} \geq m-O(\delta),~~~~
		\sum b_i \geq m-O(\delta).
	\end{equation}
	
	However, $(a_i,b_i)$ must be bounded away from $(1,1)$:
	\begin{equation}
		(1-a_i)+(1-b_i)>c>0,
	\end{equation}
	where the bound $c$ depends only on the depth $t$.
	(In fact, due to the data processing inequality, $\sqrt{a(1-b)}+\sqrt{b(1-a)}\geq 2^{-2^{t-1}}$, as in \cref{eq:DPIF}.)
	This implies $\sum (a_i+b_i)< (2-c)m$, a contradiction to \cref{eq:totalerror}.
	
	\paragraph{Technical Comments}
	One can also derive a no-go result for $\mathsf{FDU}\circ\mathsf{Clifford}$ following the beautiful arguments in \cite{korbany2025long,parham2025quantum}, based on the fact that the entropic quantities in stabilizer states are integers.
	For the ZX-cat state, $I(A:A')= 1-e^{-\Omega(|A|)}$ (\cref{eq：MI} is for $|A|=1$), where the region size $|A|$ is tied to the circuit depth in the entropy-based analysis.
	Therefore, a direct application of this method does not yield a constant approximation error tolerance.
	Although one could engineer a variant state $\ket{\psi'}\propto 2\ket{0^n}+\ket{+^n}$ to produce manifestly non-integer entropies,
	the dimension-dependence in the entropy continuity bounds would still introduce a depth-dependent overhead in the error analysis.
	
	\section{Nonabelian Topological Order}
	The no-go result for the preparation of (at least some) nonabelian topological orders through $\mathsf{FDU}\circ\mathsf{Clifford}$ has been established in the previous literature \cite{wei2025long,parham2025quantum}.
	(More precisely, these works show that if a topological code state can be prepared by $\mathsf{FDU}\circ\mathsf{Clifford}$, then the whole code space is FDU-equivalent to a stabilizer code. 
	One can then apply the belief that stabilizer codes can only realize abelian topological orders, or, in a more limited but rigorous approach, invoke constraints from code space dimensions.)

	In this section, we prove the no-go result for $\mathsf{Clifford}\circ\mathsf{FDU}$.
	
	We will consider concrete models such as quantum double \cite{kitaev2003fault} and string-net models \cite{levin2005string} on a lattice $\Lambda$ embedded in closed surfaces $\Sigma$ with genus $g\geq 1$.
	We denote the length scale of the lattice as $L$; for fixed $g$, $L=\Theta(\sqrt{n})$ where $n$ is the number of qudits.
	In such models, the Hamiltonian is a sum of commuting projectors:
	\begin{equation}\label{eq:commproj}
		H=\sum_i h_i,~~[h_i,h_j]=0,~~h_i^\dagger=h_i,~~ h_i^2=h_i,
	\end{equation}
	and ground states are frustration-free.
	The ground state subspace $\calC$ is degenerate ($\dim\calC\geq 2$) and is a quantum error-correcting code \cite{cui2020kitaev,qiu2020ground}.
	
	We focus on these specific models for convenience and rigor. 
	Our key arguments are motivated by the emergent macroscopic description, for which the correspondence with the underlying microscopic details is better understood in these models.

	\begin{definition}[topologically transversal gate, TTG]
		A topologically transversal gate is a logical gate $U$ such that $U=(\otimes_{j\notin D} U_j) \otimes U_D$ where $D$ is a topological disk and where $U_D$ is a unitary on $D$.
	\end{definition}
	It is perhaps more natural to define a topologically transversal gate as a logical gate $U$ such that there exists a decomposition of $\Lambda$ into disjoint unions of topological disks $\Lambda=\cup D_j$, such that $U=\otimes_j U_{D_j}$, where $U_{D_j}$ is a unitary on $D_j$.
	However, in our proof, we only need the case where there is at most one nontrivial disk: $\Lambda=D\cup (\cup_{i\notin D}\{i\})$.
	
	Our main result in this section is the following:
	\begin{theorem}\label{thm:nonTO}
		For quantum double models and string-net models
		that admits no TTG,
		no ground states can be prepared by $\mathsf{Clifford}\circ\mathsf{FDU}$ within an approximation error $n^{-\alpha} ~(\alpha>\frac{1}{2})$.
	\end{theorem}

	\subsection{Proof of \Cref{thm:nonTO}}
	The overall strategy is simple to state.
	Proving $\ket{\psi}\neq CU\ket{0^n}$, where $C$ is Clifford and $U$ is FDU, is equivalent to showing that $C^\dagger\ket{\psi}$ is long-range entangled for any $C$, which would follow if $C^\dagger\ket{\psi}$ is an (approximate) error-correcting code state with some nice code properties \cite{bravyi2006lieb,anshu2020circuit,yi2024complexity,bravyi2025much,yi2025lov}.

	The formal proof rests on the following two lemmas.
	\begin{lemma}\label{nologicalPauli}
		Suppose $\ket{\psi_1}$ and $\ket{\psi_2}$ are two orthogonal logical states for the above specified models that admit no TTG, then $\abs{\mel{\psi_1}{U}{\psi_2}} = \exp(-\Omega(L))$ for any Pauli operator $U=\otimes_i U_i$.
	\end{lemma}
	
	For general systems, there might be some ambiguities in defining Pauli and Clifford.
	(For example, suppose the system are made of 4-dimensional qudits, should we use $\mathbb{Z}_4$-Pauli or $\mathbb{F}_4$-Pauli?)
	To mitigate the ambiguity, we will only use the most common properties of Pauli operators: they are onsite and discrete.
	Then this lemma could be false if the code admits a transversal gate that rotates $\ket{\psi_2}$ to $\ket{\psi_1}$. 
	This is why we need to rule out transversal gates.

	\begin{lemma}\label{aqeccomplexity}
		If $\ket{\phi}$ and $\ket{\phi'}$ satisfy $\braket{\phi}{\phi'}=0$ and $\norm{\phi_A-\phi'_A}_1<\epsilon$ for any subset $A$ such that $|A|<d$, then
		\begin{equation}
			\cc_\delta(\phi) > 
			\log\frac{d}{\max\{n(\delta+\epsilon), 1\}}.
		\end{equation}
	\end{lemma}
	This lemma is a robust version of the well-known fact that quantum error correcting codes are entangled \cite{bravyi2006lieb}.
	
	Now, let us derive \cref{thm:nonTO} from \cref{nologicalPauli,aqeccomplexity}.
	\begin{proof}[Proof of \cref{thm:nonTO}]
		Suppose $\ket{\psi}\approx CU\ket{0}$ up to $\delta$ error, where $C$ is Clifford. Our task it to lower bound the depth of $U$.
		
		Let us pick another code state $\ket{\psi'}$ such that $\braket{\psi}{\psi'}=0$.
		Consider $C^\dagger\ket{\psi}$ and $C^\dagger\ket{\psi'}$.
		We claim that these two states are approximately locally indistinguishable.
		To see it, let us compare the expectation value of any local operator $V$.
		We decompose it into sum of Pauli strings $V=\sum \lambda_i P_i$ (the number of terms is at most $4^{|\supp(V)|}$).
		For any Pauli term $P$ in the expansion, note that $CPC^\dagger$ is also Pauli, hence:
		\begin{equation}
			\begin{aligned}
				\mel{\psi}{CPC^\dagger}{\psi} -\mel{\psi'}{CPC^\dagger}{\psi'} 
				=\mel{\chi_1}{CPC^\dagger}{\eta_2} + \mel{\chi_2}{CPC^\dagger}{\eta_1}
				=\exp(-\Omega(L)).
			\end{aligned}
		\end{equation}
		Here we have decomposed $\ket{\psi}$ and $\ket{\psi'}$ as $\frac{1}{\sqrt{2}}(\ket{\chi_1}\pm \ket{\chi_2})$ and applied \cref{nologicalPauli} to $\ket{\chi_{1,2}}$. 
		As shown in \cref{lem:twostab}(2), as long as $|\supp(V)|=O(L)$ for a small enough proportional constant, the total difference will be $\exp(-\Omega(L))$.
		
		So $C^\dagger\ket{\psi}$ and $C^\dagger\ket{\psi'}$ satisfy \cref{aqeccomplexity} for $\epsilon=\exp(-\Omega(L))$ and $d=\Omega(\sqrt{n})$.
		Pick $\delta=n^{-\alpha} ~(\alpha>\frac{1}{2})$ and applying the lemma, we find $\cc_\delta(\psi)=\Omega(\log n)$.
	\end{proof}
	
	\begin{proof}[Proof of \cref{nologicalPauli}]
		Let us consider the number of excitations produced by $U$, denoted by $N_e$.
		Here, an excitation is defined as a term in the Hamiltonian that does not annihilate $U\ket{\psi}$.
		Since $\calC$ is an error-correcting code and $U$ is transversal, the above notion is independent on the choice of $\ket\psi\in\calC$.
		There are two cases.
		
		\textbf{Case 1: $N_e\geq L/10$}.
		We select a collection of excitations whose supports, denoted as $A_1$, $A_2,\cdots,A_m$, are pairwise separated by a distance $\Omega(1)$.
		We can always achieve $m=\Omega(L)$ by a simple greedy algorithm.
		
		Since there are only finite number of possible Pauli matrices $U_i$ on each site, and a Hamiltonian term only involve constant number of sites, we have, for each $i$,
		\begin{equation}
			F((\psi_1)_{A_i}, (U\psi_2)_{A_i}) <1-c.
		\end{equation}
		Here $F(,)$ is the fidelity, $c$ is an absolute constant that is independent of $i$ (only depends on the model and the choice of Pauli matrices).
		Denote $A=\cup A_i$.
		Due to the disentangling lemma \cite{bravyi2010tradeoffs}, $(\psi_1)_A=\otimes (\phi_1)_{A_i}$ and $(U\psi_2)_A=\otimes (U\phi_2)_{A_i}$, hence:
		\begin{equation}
			|\mel{\psi_1}{U}{\psi_2}| \leq \prod_{i=1}^m F((\psi_1)_{A_i}, (U\psi_2)_{A_i}).
		\end{equation}
		So $\abs{\mel{\psi_1}{U}{\psi_2}} = \exp(-\Omega(m)) = \exp(-\Omega(L))$.

		\textbf{Case 2: $N_e<L/10$}.
		In this case, we claim that $\mel{\psi_1}{U}{\psi_2}=0$.
		
		Since $N_e<L/10$, we can always choose two wide enough stripes\footnote{Here we use torus as an example. The generalization to higher genus is straightforward.}, $S_v$ and $S_h$ that avoid all excitations, and all excitations are in a rectangular region $\mathring{T}$, the interior of $T$. See \cref{fig:stripes}.
		
		\begin{figure}
			\centering
			\begin{tikzpicture}
				\draw[thick] (0,0) rectangle (5,5);
				
				\def\x{1}
				\def\y{0.7}
				\draw[thick, fill=yellow] (\x,0) -- (\x+\y,0) -- (\x+\y,5-\x-\y) -- (5,5-\x-\y) -- (5,5-\x) -- (\x+\y,5-\x) -- (\x+\y,5) -- (\x,5) -- (\x,5-\x) -- (0,5-\x) -- (0,5-\x-\y) -- (\x,5-\x-\y) -- cycle;
				
				\node[scale=1.5] at (\x+\y/2,5-\x-\y/2) {$S$};
				
				\draw[thick] (2.5, 0.5) rectangle (4.5, 2.5);
				\draw[thick] (2.8, 0.8) rectangle (4.2, 2.2);
				
				\def\z{0.09}
				\fill[red] (3.2, 1.7) circle (\z);
				\fill[red] (3.3, 1.3) circle (\z);
				\fill[red] (3.7, 1.8) circle (\z);
				\fill[red] (3.8, 1.4) circle (\z);
				
				\node[scale=1] at (3.6,2.35) {$\partial T$};
				\node[scale=1] at (3.9,1.9) {$\mathring{T}$};
			\end{tikzpicture}
			\caption{Yellow stripe $S=S_v\cup S_h$. 
				Red dots denote excitations in $U\ket{\psi}$. $T=\partial T\cup \mathring{T}$.}
			\label{fig:stripes}
		\end{figure}

		We will remove the excitations and return a pure logical state by applying a unitary $R_T$ supported on $T$ (in fact, on $\mathring{T}$)\footnote{We are not quite performing a quantum error correction. In particular, $R_T$ could depend on $U$, but this is allowed since $U$ is fixed.}.
		To see such $R_T$ exists, we first show that the total anyon charge in $T$ must be deterministic (no ``superposition of anyon types").
		This is because the anyon charge can be measured by a ribbon operator in $\partial T$, as well as a ribbon operator outside $T$.
		On one hand, since two ribbon operators are spatially separated, they do not have any correlation.
		On the other hand, they are measuring the same anyon charge.
		The only consistent possibility is that both measurements yield the same, deterministic value.
		
		Due to the code property, this total anyon charge is independent of the specific logical state $\ket\psi\in\calC$.
		Furthermore, we may assume that this charge is 1, the vacuum.
		(Otherwise, $\mel{\psi_1}{U}{\psi_2}=0$ already holds.)
		We therefore conclude that, in the state $U\ket\psi$, the annulus $\partial T$ contains no excitations and encloses the trivial anyon charge, as in the original state $\ket{\psi}$.
		Therefore, 
		\begin{equation}\label{eq:equalpartialT}
			(U\ketbra{\psi}U^\dagger)_{\partial T} = (\ketbra{\psi})_{\partial T}.
		\end{equation}
		
		Now we can ``glue" the $T$-portion of $\ket{\psi}$ onto of $U\ket{\psi}$:
		we claim there exists a unitary $R_T$ such that the state $\ket\Psi \defeq R_T U \ket{\psi}$ satisfies:
		\begin{equation}\label{eq:glueresult}
			(\ketbra{\Psi})_T = (\ketbra{\psi})_T,~~~
			\Tr_{\mathring{T}}\ketbra{\Psi} = \Tr_{\mathring{T}}\ketbra{U\psi}.
		\end{equation}
		This is simple, but we defer the proof to \cref{glue} to avoid distraction.
		Eq.~(\ref{eq:glueresult}) implies $\ket{\Psi}\in\calC$.
		Moreover, we note that $R_T$ can be determined solely from $(\ketbra{\psi})_T$, hence only depends on logical space $\mathcal{C}$.
		Namely, the same $R_T$ works for all $\ket{\psi}\in\calC$.
		Therefore, the unitary $V\defeq R_TU$ is a logical gate, 
		which is topologically transversal by construction.
		By the no TTG assumption, $V$ must be the logical identity, hence
		\begin{equation}\label{eq:U=RT}
			U\ket{\psi} = R_T^\dagger \ket{\psi},~~\forall \ket\psi\in\calC.
		\end{equation}

		Now back to the mutually orthogonal states $\ket{\psi_{1,2}}$. 
		Eq.~(\ref{eq:U=RT}) implies $\ketbra{U\psi_2}$ and $\ket{\psi_2}$ are identical outside $T$, which immediately implies the desired result $\mel{\psi_1}{U}{\psi_2}=0$.
		More precisely, for such topological codes, $\braket{\psi_1}{\psi_2}$ implies the existence of a logical projector $X$ supported on $S$ that perfectly distinguish them:
		\begin{equation}
			X\ket{\psi_1}=\ket{\psi_1},~~~X\ket{\psi_2}=0.
		\end{equation}
		The locality of $R_T$ implies $XU\ket{\psi_2}=XR_T^\dagger\ket{\psi_2}=R_T^\dagger X\ket{\psi_2}=0$.
		Although $U\ket{\psi_2}\notin\calC$ in general, it already follows that $U\ket{\psi_2}\perp\ket{\psi_1}$.
	\end{proof}

	\begin{proof}[Proof of \cref{aqeccomplexity}]
		We first claim that $\norm{\phi-\ketbra{0^n}}_1< \delta$ implies $\delta+\epsilon > \frac{d}{n}$.
		Let us consider $H=\sum (1-\ketbra{0}_{A_i})$, where we have divided the full system as $\cup_i A_i$ where each $|A_i|<d$.
		There are $\lceil\frac{n}{d-1}\rceil$ terms, but we simply write it as $\frac{n}{d}$ in the following.
		$H$ has the unique ground state $\ket{0^n}$ with energy 0, and a spectral gap 1.
		
		Since $\norm{\phi-\ketbra{0^n}}_1< \delta$, we have $\expval{H}{\phi}<\frac{n\delta}{2d}$, hence $\expval{H}{\phi'}<
		\frac{n(\delta+\epsilon)}{2d}$.
		If, on the contrary, $\delta+\epsilon<\frac{d}{n}$, then $\expval{H}{\phi}, \expval{H}{\phi'}<\frac{1}{2}$, which implies 
		$\abs{\braket{0^n}{\phi}}^2, \abs{\braket{0^n}{\phi'}}^2>\frac12$.
		This forces $\braket{\phi}{\phi'}\neq 0$, a contradiction. 
		
		Now consider $U\phi U^\dagger$ and $U\phi' U^\dagger$ where $U$ is a unitary circuit of depth $t$ such that $d>2^t$.
		They satisfy the conditions after replacing $d$ by $\frac{d}{2^t}$.
		Hence, by the above argument, if $\norm{U\phi U^\dagger-\ketbra{0^n}}_1<\delta$, then $\delta+\epsilon > \frac{d}{2^tn}$.
		Therefore,
		\begin{equation}
			\cc_\delta(\phi) > 
			\log\frac{d}{\max\{n(\delta+\epsilon), 1\}}.\tag*{\qedhere}
		\end{equation}
	\end{proof}
	
	\subsection{Comments on Transversal Gates}
	In this subsection, we present explicit models that satisfy the no-TTG condition. 
	Both arguments are essentially based on the macroscopic description of topological order—specifically, topological quantum field theory (TQFT) and unitary modular tensor categories (UMTC). 
	The first argument is rigorous. 
	The second employs a more heuristic, physics-based perspective; it is provided for completeness and can be omitted by readers who prefer a strictly rigorous treatment.
	
	\subsubsection{Via Universality of Mapping Class Group}
	
	We begin by very briefly summarizing relevant arguments from Ref.~\cite{beverland2016protected} (particularly Sections 3.1, 3.2, and 4.1).
	This work concerns locality-preserving logical gates (henceforth LPU gates, where ``U" stands for unitary), which include transversal gates in particular.
	
	Denote the algebra of linear maps $\calC\to\calC$ as $[\calA_\calC]$.
	A logical gate $L$ induces a well-defined automorphism $\rho_L$ of $[\calA_\calC]$ as:
	\begin{equation}\label{eq:rhoU}
		\rho_L: [X] \mapsto L^\dagger [X] L = [L^\dagger X L],
	\end{equation}
	where $X:\mathcal{H}\to \mathcal{H}$ is a $\calC$-preserving map, $[X]: \calC\to\calC$ is its restriction on $\calC$.
	Given a nontrivial loop $C$ on the surface $\Sigma$, consider the subalgebra $[\calA(C)]\subseteq [\calA_\calC]$ consisting of elements that can be realized by operators along $C$. 
	It is isomorphic to the Verlinde algebra, and admits a unique basis of minimal idempotents such that the algebra decomposes as $\bigoplus \mathbb{C}$.
	If $L$ is locality-preserving, then $\rho_L$ restricts to an automorphism of $[\calA(C)]$, which acts as a permutation on this idempotent basis.
	
	We employ a set of loops to to define a pants decomposition of $\Sigma$.
	This decomposition defines a basis for $\mathcal{C}$ in which the Verlinde algebras of all defining loops are diagonal. Assuming the fusion rules are multiplicity-free ($N_{ab}^c \in \{0, 1\}$), each basis vector is uniquely specified by its simultaneous eigenvalues (physically, anyon flux labels). Consequently, in this basis, 
	\begin{equation}\label{eq:Lmono}
		\text{ $L$ must be a monomial matrix\footnote{A matrix is called monomial if it has the form $\Pi D$, where $\Pi$ is a permutation matrix and $D$ is a diagonal unitary.}.}
	\end{equation}
	This property holds for any valid pants decomposition of $\Sigma$.
	In particular, if a mapping class group (MCG) element, which transforms one pants decomposition to another, acts projectively on $\mathcal{C}$ via the unitary $U$ (effectively the change of basis matrix), then the transformed gate $U L U^\dagger$ must also be monomial.
	
	The observation (\ref{eq:Lmono}) (and that for $U L U^\dagger$) is all we need from Ref.~\cite{beverland2016protected}.
	We note that while Ref.~\cite{beverland2016protected} considers LPU gates, the same arguments apply to our topologically transversal gates, since nontrivial loops can be deformed to avoid the potentially problematic disk-like region.

	\begin{proposition}\label{thm:doubleuniversal}
		On a closed surface, if the mapping class group for a TQFT corresponding to a UMTC $\calT$ is universal, and $\calT$ has multiplicity-free fusion rule, then the TQFT corresponding to $\calT \boxtimes\Bar{\calT}$ has no nontrivial LPU gates.
	\end{proposition}
	It is known that some TQFTs have multiplicity-free fusion rule and universal MCG on closed surfaces of genus $g\geq 2$ (for example, those in Ref.~\cite{freedman2001two,larsen2005density}).
	Taking any such theory and its time reversal copy, the total system can be realized via the string-net construction: $\mathcal{Z}(\calT) \cong \calT \boxtimes \overline{\calT}$ (note that $\calT$ is already modular) \cite{muger2003subfactors}, and our proposition then applies.
	
	We note that, while no topological orders can have universal MCG on torus, some ``doubled" models on torus still admits no LPU gates, see appendix \ref{sec:doubleFib}.

	\begin{proof}[Proof of proposition \ref{thm:doubleuniversal}]
		Denote the ground state subspace of theory $\calT$ as $\calC$.
		For $\calT \boxtimes\Bar{\calT}$, the ground state subspace is $\calC\otimes \calC^\ast$, and the mapping class group (MCG) acts on it by $\rho_g=U(h)\otimes U^\ast(h)$.
		
		We claim that the universality implies $\dim\calC\geq 3$.
		In fact, the modular representation for any MTC has finite image in $PU(\calC)$ \cite{ng2010congruence}, so we are forced to consider surface with genus $g\geq 2$.
		Then $\dim\calC$ can be computed by Verlinde's formula \cite{verlinde1988fusion}:
		\begin{equation}
			\dim\calC = \sum_i (\frac{\mathcal{D}}{d_i})^{2g-2}>\rank(\calT)\geq 2.
		\end{equation}
		Here $d_i$ is the dimension of each anyon, $\mathcal{D}$ is the total quantum dimension, $\rank(\calT)$ equals the number of anyon types.
		
		Suppose $L$ is a LPU gate, which is monomial, let us consider $\rho_g L\rho_g^\dagger$ for $\forall g\in \text{MCG}$, which must also be monomial, too.
		Since (the image of) $\{U(h)\}$ is dense in $PU(\calC)$, by continuity, $(U\otimes U^\ast)L(U\otimes U^\ast)^\dagger$ must be monomial for $\forall U\in U(\calC)$.
		\Cref{lemma:doubleuniversal} below then implies that $L$ is identity.
	\end{proof}
	
	\begin{lemma}\label{lemma:doubleuniversal}
		Let $K \in M_{n^2}(\mathbb{C})$, $n\geq 3$. If the matrix $K_U = (U \otimes U^\ast) K (U \otimes U^\ast)^\dagger$ is a monomial matrix for every unitary $U \in U(n)$, then $K = \lambda I ~(\lambda\in\mathbb{C})$.
	\end{lemma}
	\begin{proof}
		Denote the base of $\calC$ as $\{e_i\}$.
		For any $x\in\calC$, we choose a unitary $U$ such that $x=U^\dagger e_1$.
		Then $ K(x\otimes x^\ast) = K(U^\dagger \otimes U^T) (e_1\otimes e_1^\ast) = (U \otimes U^\ast)^\dagger K_U (e_1\otimes e_1^\ast)$, which has the form of $y_1\otimes y_2$.
		
		We claim that this implies that $K$ sends all product states to product states.
		Let us define a map $f: \mathbb{C}^{n}\times \mathbb{C}^{n} \to M_{n^2}(\mathbb{C})$ as:
		\begin{equation}
			f(z,w)=K(z\otimes w),
		\end{equation}
		where we view the r.h.s. as an $n\times n$ matrix.
		$K(z\otimes w)$ has tensor product form if and only if its rank (as a matrix) is 1, if and only if the determinants of all $2\times 2$ submatrices vanish.
		For any $2\times 2$ determinant $g$, consider $g\circ f: \mathbb{C}^{n}\times \mathbb{C}^{n} \to\mathbb{C}$, which is a quadratic polynomial of $z$ and $w$.
		We already proved that $(g\circ f)(z,\bar{z})=0,~~\forall z\in\mathbb{C}^n$. 
		Define $p: \mathbb{C}^{n}\times \mathbb{C}^{n} \to\mathbb{C}$ as $p(x,y)=(g\circ f)(x+iy,x-iy)$, then $p(x,y)=0,~~\forall x,y\in\mathbb{R}^n$.
		This implies $p$ is a zero polynomial, so does $g\circ f$.
		This holds for any $2\times 2$ determinant $g$, hence $K(z\otimes w)$ has tensor product from for any $z$ and $w$.

		Hence $K$ has the form of $A\otimes B$ or $(A\otimes B)\text{SWAP}$ \cite{westwick1967transformations}.
		
		(1) $K=A\otimes B$. Then $K_U = U A U^\dagger \otimes  U^\ast B U^T$.
		Without loss of generality, we can assume $A$, $B$, $U A U^\dagger$, and $U^\ast B U^T$ are all monomial\footnote{It is obvious that $X\otimes Y$ is monomial if and only if there exists $\lambda\in\mathbb{C}$ such that both $\lambda X$ and $\lambda^{-1}Y$ are monomial.}.
		We claim that both $A$ and $B$ are scalar matrices.
		In fact, $A\mapsto UAU^\dagger$ is the representation $\textbf{n}\otimes\bar{\textbf{n}}$, which decomposes to $\textbf{1}\oplus\textbf{Adj}$.
		Due to continuity, $U A U^\dagger$ must have the same permutation type as $A$, hence is contained in an $n$-dimensional subspace of $\textbf{n}\otimes\bar{\textbf{n}}$.
		Projecting this subspace to \textbf{Adj}, the dimension is still at most $n$.
		However, if $A$ has any component in \textbf{Adj}, due to irreducibility, such projected space should have dimension $n^2-1>n$, a contradiction\footnote{The argument here gives an alternative proof of Corollary 4.2 of Ref.~\cite{beverland2016protected}.}.
		
		(2) $K=(A\otimes B)\text{SWAP}$. Then $K_U = (U A U^T\otimes U^\ast B U^\dagger) \text{SWAP}$.
		Similarly, we can assume $A$, $B$, $U A U^T$, and $U^\ast B U^\dagger$ are all monomial.
		We claim that such $A$ (and $B$) cannot exist.
		In fact, $A\mapsto UAU^T$ is the representation $\textbf{n}\otimes\textbf{n}$, which decomposes to $\textbf{Sym}\oplus\textbf{Alt}$.
		By the same argument as in (1), since $\dim(\textbf{Sym})=\frac{n(n+1)}{2}>n$, $A$ cannot have any component in \textbf{Sym}.
		Hence $A$ is anti-symmetric.
		If $n\geq 4$, then $\dim(\textbf{Alt})=\frac{n(n-1)}{2}>n$, implying that $A$ does not exist.
		If $n=3$, then it is straightforward to check that anti-symmetric $3\times 3$ monomial matrix does not exist either.
	\end{proof}
	
	We remark in passing that, rather than using \cref{thm:doubleuniversal}, it would be easier to directly apply corollary 4.2 of Ref.~\cite{beverland2016protected} (that universality of MCG implies no LPU gates) to rule out LPU gates if one have a quantum double or string-net model with universal MCG.
	However, this is impossible.
	\begin{proposition}\label{thm:nonchiraluniversal}
		If $\calT$ is equivalent to $\bar{\calT}$, then the MCG on a closed surface cannot be universal.
	\end{proposition}
	\begin{proof}[Proof sketch]
		Denote the ground state subspace of theory $\calT$ as $\calC$.
		Similar to \cref{thm:doubleuniversal}, universality implies $\dim\calC\geq 3$.
		The theory $\calT$ defines a representation of the MCG, denoted by $\rho$.
		The representation corresponding to $\bar{\calT}$ is $\bar{\rho}$.
		If $\calT\cong\bar{\calT}$, then $\rho\cong\bar{\rho}$.
		Namely, $\rho$ is pseudoreal: there exists a unitary $V$, such that
		\begin{equation}
			U(h) \propto V U^\ast(h) V^\dagger,~~\forall h\in\text{MCG}. 
		\end{equation}
		Hence $U(h)VU^T(h)\propto V$.
		Universality of $\rho$ and continuity implies that $UVU^T\propto V$ for $\forall U\in U(\calC)$.
		However, $\textbf{n}\otimes\textbf{n}=\textbf{Sym}\oplus\textbf{Alt}$ does not contain any 1-dimensional representation when $n\geq 3$.\footnote{
			One can also view \cref{thm:nonchiraluniversal} as a corollary of \cref{thm:doubleuniversal}, since $(V\otimes V^\dagger)\text{SWAP}$ would be a nontrivial LPU gate for $\calT \boxtimes\Bar{\calT}$.
		}
	\end{proof}

	\subsubsection{Via Symmetries in TQFT}
	We provide a physical, TQFT-based argument demonstrating that LPU gates cannot exist in certain quantum double models. In the language of TQFT, locality preserving logical gates correspond to symmetries of the underlying theory. In two spatial dimensions, such a symmetry can be either a 0-form or a 1-form symmetry.
	
	1-form symmetries manifest as string-like logical operators, corresponding to moving an anyon along a non-trivial loop. 
	We need these anyons be Abelian, as non-Abelian string operators require linear depth to implement \cite{bravyi2022adaptive}.
	For quantum doubles $D(G)$, non-trivial Abelian anyons correspond to non-trivial central elements of $G$, non-trivial one-dimensional representations of $G$, or both \cite{kitaev2003fault}. 
	Both are forbidden if $G$ is nonabelian and simple.
	
	0-form symmetries correspond to anyon permutations, or braided autoequivalences in the category language.
	(Physically this correspondence is realized by ``sweeping the domain wall" \cite{yoshida2015topological,yoshida2017gapped}.)
	For quantum doubles $D(G)$, these are described by the Brauer-Picard group of $G$ \cite{nikshych2014categorical}.
	It is known that $\text{BrPic}(G)$ can be trivial for several simple groups.
	
	To summarize, it suffices to select a simple group $G$ with trivial Brauer-Picard group to exclude these symmetries.
	Concrete examples include the Mathieu group $M_{11}$ and the Monster group \cite{nikshych2014categorical}.

	\section{Discussion and Outlook}
	
	We introduce two-sided long-range magic, a robust form of nonstabilizerness defined by a state's inability to be prepared by Clifford augmented shallow circuits ($\mathsf{FDU}\circ\mathsf{Clifford}$ or $\mathsf{Clifford}\circ \mathsf{FDU}$). 
	We establish this property in concrete and physically distinct examples, including the ZX-cat state and ground states of certain nonabelian topological orders. 
	These results expand the known boundary of shallow state-preparation complexity and forge a rigorous connection between computational complexity and the intrinsic structure of topological phases.

	Looking ahead, our work motivates several intriguing directions for future research:
	
	\textbf{Higher Levels of Magic Hierarchy}.
	A compelling perspective is that the magic hierarchy defines of a stronger many-body magic: if we postulate that both Clifford circuits and FDUs do not promote short-range magic to long-range magic, then the entire magic hierarchy only produce short-range magic as well.
	
	Our results in this work concern the first level of the magic hierarchy. 
	A natural next step is to construct explicit states and prove they are outside higher levels.
	In appendix \ref{sec:QCQ}, we show that a ZX-cat-like state in fact lies in the second level of the hierarchy.
	For nonabelian topological orders, the macroscopic nature of our TQFT-based arguments suggests a generalized proof against $\mathsf{FDU} \circ \mathsf{Clifford} \circ \mathsf{FDU}$ may be feasible.
	However, progressing substantially further will likely require new tools.
	
	\textbf{The Landscape of Topological Magic}.
	Our proof for nonabelian topological orders leverages the (approximated) error-correcting code property and hence requires ground state degeneracy.
	We conjecture that the same conclusion holds for the states on the sphere.
	Proving this would likely necessitate a more intrinsic approach, directly utilizing the modular data.
	Similar story already exists for long range entanglement \cite{haah2016invariant,aharonov2018quantum,li2025much}.
	Such proof might also relax the no TTG condition and improve the tolerated approximation error.
	
	\textbf{Adaptive Circuits}.
	A key motivation for studying the magic hierarchy is its connection \cite{parham2025quantum} to adaptive circuits (those with measurement and feedforward), which are known to be powerful~\cite{verresen2021efficiently,tantivasadakarn2024long,bravyi2022adaptive,lu2022measurement}.
	For the states considered in this work, an interesting divergence emerges.
	The ZX-cat state can be prepared efficiently using constant-depth adaptive circuit (see appendix \ref{sec:QCQ}).
	In contrast, the adpative preparation of general nonabelian orders are conjectured to be hard \cite{tantivasadakarn2023hierarchy}.
	Rigorously characterizing which states are resilient to adaptive preparation remains an open problem.

	\textbf{Connections to Fault Tolerance}.
	A recurring theme in quantum computation is the trade-off between fault-tolerant gate implementability and computational universality, exemplified by theorems such as Eastin–Knill \cite{eastin2009restrictions} and Bravyi-König \cite{bravyi2013classification}. 
	Our work provides a concrete manifestation of this principle: the absence of topologically transversal gates, a condition regarding fault-tolerant operations, enforces the presence of two-sided long-range magic, a form of non-Clifford computational resource.
	This interplay between the long-range magic and fault-tolerant structures clearly warrants further investigation.

	\textbf{Quantitative Measures and Unified Characterizations}.
	At present, our proofs are tailored to specific states and rely on different mechanisms in different settings. It would be desirable to develop a more unified framework capable of diagnosing two-sided long-range magic in a unified way. 
	Candidate approaches include developing quantitative measures that detect long-range magic directly---such as magic monotones or other types of invariants.

	\section*{Acknowledgement}
	I thank Michael Beverland for pointing out Ref.~\cite{beverland2016toward}, which is helpful in appendix \ref{sec:doubleFib}.
	I thank Sergey Bravyi for commenting on the adaptive preparation of ZX-states and pointing out Ref.~\cite{kim2025any}.
	I thank Robert König, Hari Krovi, Tsung-Cheng Lu, and Beni Yoshida for helpful discussions.
	Part of the proof of \cref{nologicalPauli} overlaps conceptually with a proof in a forthcoming collaborative work \cite{chiralpaper}. I thank my collaborators for agreeing to the prior release of this work.    
	
	\begin{appendix}
		
		\section{ZX-cat State: Alternative Proof}\label{sec:ZX2ndproof}
		In this section, we provide an alternative proof for preparing the ZX-cat state by $\mathsf{Clifford}\circ\mathsf{FDU}$.
		For the simplest case where the Clifford is trivial, we can observe that $\ket{\psi}$ is approximately locally indistinguishable with another ZX-cat-like state,
		\begin{equation}\label{eq:psiprime}
			\ket{\psi'}=\frac{1}{\sqrt{2\beta}}(\ket{0^n}-\ket{+^n}),~~~\beta=1-\frac{1}{2^n}.
		\end{equation}
		Hence, they form an approximated ``code", and hence requires nontrivial circuit depth to prepare. 
		We use quotation mark since they do not quite form an approximated code in the sense that the whole code space are locally indistinguishable.
		(In fact, $\ket{0^n}, \ket{+^n}\in\mathrm{span}\{\ket\psi,\ket{\psi'}\}$ are locally distinguishable).
		However, our \cref{aqeccomplexity,ccagsp} do apply.
		
		\subsection{Improved Bound for Approximated ``Codes" with Large Distances}
		\begin{lemma}\label{ccagsp}
			Suppose $\ket{\phi}$ and $\ket{\phi'}$ satisfy\footnote{The condition $\braket{\phi}{\phi'}=0$ can also tolerate constant error.} $\braket{\phi}{\phi'}=0$ and $\norm{\phi_A-\phi'_A}<\epsilon$ for any subset $A$ such that $|A|<d$, where $d=\Omega(n^{\frac{1}{2}+\alpha})~(\alpha>0)$ and $\epsilon=\exp(-\Omega(d))$.
			Then for a small enough constant $\delta>0$, we have
			\begin{equation}
				\cc_\delta(\phi)=\Omega(\log n).
			\end{equation}
		\end{lemma}
		The proof is based on the approximated ground state projector (AGSP) method, inspired by lemma 14 of Ref.~\cite{anshu2020circuit}.
		
		\begin{proof}
			Consider $G=\sum_{i=1}^n G_i$ where $G_i=\ket{1}\bra{1}_i$. Its spectrum is $\{0,1,\cdots,n\}$, and $\ket{0^n}$ is the unique ground state, hence
			\begin{equation}
				\ket{0^n}\bra{0^n}=P_0(G),
			\end{equation}
			where $P_0:\{0,1,\cdots,n\}\to \{0,1\}$ is a ``step function": $P_0(0)=1$ and $P_0(i)=0$ for $i>0$. 
			We use a polynomial $P(x)$ of degree $m< d$ (specified later) to approximate $P_0(x)$.
			For our purpose, the following polynomial in Ref.~\cite{arad2013area} suffices and is convenient:
			\begin{equation}
				P(x)=\frac{T_m(\frac{n+1-2x}{n-1})}{T_m(\frac{n+1}{n-1})},
			\end{equation}
			where $T_m(x)$ is the degree-$m$ Chebyshev polynomial. 
			It satisfies $P(0)=1$ and $|P(x)-P_0(x)|\leq 2e^{-\frac{2m}{\sqrt{n}}}$ for $x\in [1,n]$ \cite{arad2013area}, hence
			\begin{equation}\label{eq:AGSP}
				\norm{\ket{0^n}\bra{0^n}-P(G)}_\infty \leq 2e^{-\frac{2m}{\sqrt{n}}}.
			\end{equation}
			
			Now, for the given $\forall\ket{\phi}, \ket{\phi'}$, we have:
			\begin{equation}
				\begin{aligned}
					\expval{\phi|P(G)|\phi} - \expval{\phi'|P(G)|\phi'}
					&=\sum_{k=0}^m a_k \left(
					\mel{\phi}{(\sum_{i=1}^n G_i)^k}{\phi}-\mel{\phi'}{(\sum_{i=1}^n G_i)^k}{\phi'}
					\right)
					\leq \frac{1}{2}\sum_{k=0}^m   |a_k|n^k\epsilon     .
				\end{aligned}
			\end{equation}
			The last step is because there are $n^k$ terms after expanding $(\sum_{i=1}^n G_i)^k$ and each term is a projector supported on at most $k$ sites.
			Combining it with \cref{eq:AGSP}, we get:
			\begin{equation}\label{eq:FTTboundviaAGSP}
				\abs{\expval{\phi|0^n}}^2 - \abs{\expval{\phi'|0^n}}^2 
				=  \expval{\phi|0^n}\expval{0^n|\phi}- \expval{\phi'|0^n}\expval{0^n|\phi'}
				\leq 
				4e^{-\frac{2m}{\sqrt{n}}}+\frac{1}{2}\sum_{k=0}^m |a_k|n^k\epsilon.
			\end{equation}
			All roots of $T_d(x)$ lie in $[-1,1]$, hence all roots of $P(x)$ are positive, hence the coefficients of $P(x)$ alternate signs. 
			Therefore,
			\begin{equation}
				\sum_{k=0}^m |a_k|n^k
				=\abs{P(-n)}
				= \abs{\frac{T_m(\frac{3n+1}{n-1})}{T_m(\frac{n+1}{n-1})}}
				= e^{O(m)}.
			\end{equation}
			The last step can be verified by the explicit formula $T_m(x)=\cosh(m\,\mathrm{arccosh}(x))$ for $x\geq 1$.
			
			We pick $m=\min\{d-1,c\log\frac{1}{\epsilon}\}$, where $c$ is a small enough absolute constant so that $e^{O(m)}\epsilon$ is $O(\sqrt{\epsilon})$, hence is still $\exp(-\Omega(d))$.
			It follows that:
			\begin{equation}\label{eq:agspphipsi}
				\abs{\expval{\phi|0^n}}^2 - \abs{\expval{\phi'|0^n}}^2
				\leq 
				4e^{-\frac{2m}{\sqrt{n}}} + O(\sqrt{\epsilon}) \lesssim e^{-\frac{2m}{\sqrt{n}}}.
			\end{equation}
			Since $\braket{\phi}{\phi'}=0$, we have $\abs{\expval{\phi'|0^n}}^2 + \abs{\expval{\phi|0^n}}^2\leq 1$.
			Combining it with \cref{eq:agspphipsi}, we get:
			\begin{equation}
				\abs{\expval{\phi|0^n}}^2 \leq \frac{1}{2} + O(e^{-\frac{2m}{\sqrt{n}}}).
			\end{equation}
			
			Now we conjugate the subspace $\calC$ by a depth-$t$ unitary circuit $U^\dagger$ and repeat the above argument (namely, replace $d$ by $\frac{d}{2^t}$).
			We get:
			\begin{equation}\label{eq:agspphiU0}
				\abs{\mel{\phi}{U}{0^n}}^2
				\leq \frac{1}{2} + O(        e^{-\Theta(\frac{n^\alpha}{2^{t}})}).
			\end{equation}
			The l.h.s. should be greater than $1-\delta$ by assumption, hence $t= \Omega(\log n)$.    
		\end{proof}
		
		\textbf{Remark}: 
		Furthermore, if there is a subspace $\calC$ ($\dim\calC\geq 2$) such that $\norm{\phi_A-\phi'_A}_1<\epsilon$ for any $\ket{\phi}, \ket{\phi'}\in \calC$ and any subset $A$ such that $|A|< d$, then \cref{eq:agspphiU0} can be improved to:
		\begin{equation}
			\abs{\mel{\phi}{U}{0^n}} \lesssim
			e^{-\Theta(\frac{n^\alpha}{2^{t}})},~~~~\forall \ket{\phi}\in\calC.
		\end{equation}
		Indeed, we can always pick a $\ket{\phi'}\in \calC$ such that $\braket{\phi'}{0^n}=0$, hence the l.h.s. of \cref{eq:agspphipsi} becomes $\abs{\expval{\phi|0^n}}^2$.
		Repeating the argument for $U^\dagger \calC$ gives the desired result.
		
		\subsection{ZX-cat States as Approximated ``Codes"}
		Recall that we have defined $\ket{\psi}$ and $\ket{\psi'}$ as in \cref{eq:psi,eq:psiprime}.
		Assuming $\ket{\psi}=CU\ket{0^n}$ by contradiction, as done in the main text, we define  $\ket{\phi} = C^\dagger \ket{\psi}$ and $\ket{\phi'}=C^\dagger \ket{\psi'}$, so
		\begin{equation}
			\ket{\phi}=\frac{1}{\sqrt{2\alpha}}(\ket{\phi_1}+\ket{\phi_2}),~~~~
			\ket{\phi'}=\frac{1}{\sqrt{2\beta}}(\ket{\phi_1}-\ket{\phi_2}).
		\end{equation}
		Our task is to show that $\ket{\psi}$ requires large circuit depth to prepare.
		
		It is clear that $\braket{\phi}{\phi'} = \braket{\psi}{\psi'} = 0$.
		Furthermore, we claim that $\ket{\phi}$ and $\ket{\phi'}$ are approximately locally indistinguishable.
		In fact, let us compute the difference of $\expval{V}{\phi}$ and $\expval{V}{\phi'}$ for a local operator $V$. 
		Using \cref{lem:twostab}, we have:
		\begin{align}
			\expval{V}{\phi} - \expval{V}{\phi'}
			=&(\frac{1}{2\alpha}-\frac{1}{2\beta})\left( \expval{V}{\phi_1}+\expval{V}{\phi_2} \right)
			+ (\frac{1}{2\alpha}+\frac{1}{2\beta})\left( \mel{\phi_1}{V}{\phi_2} +\mel{\phi_2}{V}{\phi_1} \right) \nonumber \\
			=& O(2^{a-\frac{n}{2}}),
		\end{align}
		where $a=|\supp(V)|$.
		
		Now applying \cref{ccagsp} with $d=\frac{n}{3}$ and $\epsilon=O(e^{-\frac{n}{6}})$, we get $\cc_\delta(\psi)=\Omega(\log n)$ for a constant $\delta$.
		(Alternatively, picking $d=\frac{n}{3}$ and applying \cref{aqeccomplexity}, we find $\cc_\delta(\psi)=\Omega(\log n)$ for $\delta=\frac{1}{n^\alpha}$ ($\alpha>0$).)

		\section{Preparation of ZX-cat State}\label{sec:QCQ}
		
		\subsection{2nd-Level Preparation of a ZX-cat-like State}
		We claim that the following ZX-cat-like state
		\begin{equation}
			\ket{\tilde\psi}= \frac{1}{\sqrt{2}}(\ket{0^n}+i\ket{+^n})
		\end{equation}
		can be prepared by $\mathsf{FDU}\circ\mathsf{Clifford}\circ \mathsf{FDU}$.
		
		To start, we note that
		\begin{equation}
			\ket{\tilde\psi}   = e^{i\frac{\pi}{4}H^{\otimes n}}\ket{0^n}.
		\end{equation}
		Define $U=e^{-i\frac{\pi}{8}Y}$ (which is Clifford-equivalent to a $T$ gate), then $UZU^\dagger = H$, $U^{\otimes n}Z^{\otimes n}(U^\dagger)^{\otimes n} = H^{\otimes n}$.
		Therefore,
		\begin{equation}
			e^{i\frac{\pi}{4}H^{\otimes n}} = U^{\otimes n} e^{i\frac{\pi}{4}Z^{\otimes n}} (U^\dagger)^{\otimes n}.
		\end{equation}
		The gate $C=e^{i\frac{\pi}{4}Z^{\otimes n}}$ is a Clifford gate.
		In fact, $C=\frac{1}{\sqrt{2}}(1+iZ^{\otimes n})$, hence for any Pauli $P$,
		\begin{equation}
			C P C^\dagger = \frac{1}{2}(1+iZ^{\otimes n})P(1-iZ^{\otimes n}) 
			=\begin{cases}
				P,~~~\text{if~}[P,Z^{\otimes n}]=0\\
				iZ^{\otimes n}P,~~~\text{if~}\{P,Z^{\otimes n}\}=0
			\end{cases},
		\end{equation}
		which is always a Pauli.
		
		\subsection{Constant-depth Adaptive Preparation}
		We claim that the ZX-cat state $\ket{\psi}$ can be prepared probabilistically with a success probability greater than one-half.
		
		To do so, we first adaptively prepare GHZ on ancilla, which can be done in constant depth:
		\begin{equation}
			\ket{0^n}_a\ket{0^n} \to \frac{1}{\sqrt{2}}(\ket{0^n}_a+\ket{1^n}_a)\ket{0^n}.
		\end{equation}
		Then, we use apply parallel controlled-Hadamard gates $(CH)^{\otimes n}$ and obtain:
		\begin{equation}
			\frac{1}{\sqrt{2}}(\ket{0^n}_a\ket{0^n}+\ket{1^n}_a\ket{+^n}).
		\end{equation}
		Then we measure the ancilla in $X$ basis.
		Outcomes with even parity will correspond to the ZX-cat state $\ket{\psi}$.
		
		Relatedly, the state can be prepared using a controlled-$H^{\otimes n}$ gate and measurement, the former is conjugated to the Fanout gate \cite{kim2025any}, which can be realized adaptively in constant depth.
		
		\subsection{Matrix Product State Representation}
		The ZX-cat state $\ket{\psi}$ can be easily represented as a matrix product state (MPS) with bond dimension 2.
		
		To do so, we simply define the rank-3 tensor $A^{\alpha}_{ij}$ such that:
		\begin{equation}
			A_{00} = \ket{0},~~~A_{11}=\ket{+}~~~(\text{as vectors over the physical leg $\alpha$}).
		\end{equation}
		Equivalently,
		\begin{equation}
			A^{0} = \begin{pmatrix}
				1&0\\
				0&\frac{1}{\sqrt{2}}
			\end{pmatrix}
			,~~~
			A^{1} = \begin{pmatrix}
				0&0\\
				0&\frac{1}{\sqrt{2}}
			\end{pmatrix}
			.
		\end{equation}
		It is clear that the tensor $A$ gives rise to $\ket{\psi}$ using either periodic boundary condition or open boundary condition:
		\begin{equation}
			\ket{\psi} \propto \sum_{\{a_i\}} \Tr[A^{a_1}\cdots A^{a_n}]\ket{a_1\cdots a_n}
			=   \sum_{\{a_i\}} \ell^t A^{a_1}\cdots A^{a_n} r \ket{a_1\cdots a_n},
		\end{equation}
		where $\ell=r=(1,1)^t$.
		
		The tensor $A$ satisfies the following:
		\begin{equation}\label{eq:push}
			\begin{tikzpicture}[scale = 1, baseline = {([yshift=-.5ex]current bounding box.center)}] 
				\draw[color = black, thick] (-1, 0) -- (1, 0);
				\draw[color = black, thick] (0, 0) -- (0, -0.8);
				\draw[fill = yellow!30] (0,0) circle (0.3);
				\node at (0,0) {\small $A$};
			\end{tikzpicture}
			=
			\begin{tikzpicture}[scale = 1, baseline = {([yshift=-.5ex]current bounding box.center)}] 
				\draw[color = black, thick] (-1, 0) -- (1, 0);
				\draw[color = black, thick] (0, 0) -- (0, -0.8);
				
				\draw[fill = white] (-0.6, 0) circle (0.15);
				\node at (-0.6, 0) {\tiny $Z$};
				
				\draw[fill = yellow!30] (0,0) circle (0.3);
				\node at (0,0) {\small $A$};
				
				\draw[fill = white] (0.6, 0) circle (0.15);
				\node at (0.6, 0) {\tiny $Z$};
			\end{tikzpicture}
			=
			\begin{tikzpicture}[scale = 1, baseline = {([yshift=-.5ex]current bounding box.center)}] 
				\draw[color = black, thick] (-1, 0) -- (1, 0);
				\draw[color = black, thick] (0, 0) -- (0, -0.8);
				\draw[fill = white] (-0.6, 0) circle (0.15);
				\node at (-0.6, 0) {\tiny $X$};
				\draw[fill = yellow!30] (0,0) circle (0.3);
				\node at (0,0) {\small $A$};
				
				\draw[fill = white] (0.6, 0) circle (0.15);
				\node at (0.6, 0) {\tiny $X$};
				
				\draw[fill = white] (0, -0.5) circle (0.15);
				\node at (0, -0.5) {\tiny $H$};
			\end{tikzpicture}.
		\end{equation}
		Using this equation in mind, one can also use the following protocol to prepare $\ket{\psi}$ (see \cite{smith2023deterministic,smith2024constant,sahay2025classifying,stephen2025preparing,zhang2024characterizing}).
		\begin{enumerate}[label=\arabic*), nosep]
			\item Prepare individual $A$ tensors as 3-qubit states.
			\item Connect virtual legs using Bell measurements.
			\item For Bell measurements returning $X$, $Y$ or $Z$, push the errors towards one direction using \cref{eq:push}.
			\item If all errors cancel, we are done; otherwise, restart.
		\end{enumerate}

		\section{Gluing Two States}
		\begin{lemma}\label{glue}
			If two pure states $\ket{\psi}$ and $\ket{\psi'}$ on $ABCD$ satisfy:
			\begin{align}
				&\psi_{BC} = \psi'_{BC},\label{eq:B1}\\
				&I(A,CD)_\psi=I(AB,D)_{\psi'}=0,\label{eq:B2} \\
				&S(D)_\psi = S(D)_{\psi'},\label{eq:B3}
			\end{align}
			then there exists a unitary $U_A$ such that
			the state $\ket{\Psi}\defeq U_A\ket{\psi'}$ satisfies:
			\begin{equation}\label{eq:gluereq}
				\begin{aligned}
					&\Psi_{ABC}=\psi_{ABC},~~~\Psi_{BCD}=\psi'_{BCD},\\
					&I(A,CD)_\Psi=I(AB,D)_\Psi=0.
				\end{aligned}
			\end{equation}
		\end{lemma}
		To apply this lemma to \cref{nologicalPauli}, we set $A=\mathring{T}$, $BC=\partial T$ (divide it into a inner annulus and an outer annulus), $D=T^c$, $\ket{\psi'}=U\ket{\psi}$.
		Eq.~(\ref{eq:B1}) comes from \cref{eq:equalpartialT}; \cref{eq:B2,eq:B3} come from the disentangling lemma \cite{bravyi2010tradeoffs} and the topologically transversal gate assumption.
		
		The gluing procedure is similar to that in Ref.~\cite{kato2016information}. 
		
		\begin{proof}
			$I(A,CD)_\psi=0$ implies that there exists an isometry $V_B$ such that
			\begin{equation}
				\ket{\psi} = V_B \ket{\psi}_{AB_1}\ket{\psi}_{B_2CD}.
			\end{equation}
			Similarly, define $V_C$ via $\ket{\psi'}$ such that
			\begin{equation}
				\ket{\psi'} = V_C \ket{\psi'}_{ABC_1}\ket{\psi'}_{C_2D}.
			\end{equation}
			Define $\ket{\phi}=V_B^\dagger V_C^\dagger\ket{\psi}$ and $\ket{\phi'}=V_B^\dagger V_C^\dagger\ket{\psi'}$. 
			We have:
			\begin{align}\label{eq:defphiphi'}
				\ket{\phi} = \ket{\phi}_{AB_1}\ket{\phi}_{B_2C_1C_2D},~~~
				\ket{\phi'} = \ket{\phi'}_{AB_1B_2C_1}\ket{\phi'}_{C_2D}.
			\end{align}
			Reducing to $B_1B_2C_1C_2$, we get:
			\begin{equation}
				\phi_{B_1B_2C_1C_2}=\phi_{B_1}\otimes \phi_{B_2C_1C_2},~~~
				\phi'_{B_1B_2C_1C_2}=\phi'_{B_1B_2C_1}\otimes \phi'_{C_2}.
			\end{equation}
			Note that they are equal by assumption, hence:
			\begin{equation}
				\phi_{B_1B_2C_1C_2}=\phi'_{B_1B_2C_1C_2}=\phi_{B_1}\otimes \phi_{B_2C_1}\otimes \phi_{C_2}.
			\end{equation}
			
			We claim that $\phi_{B_2C_1}$ is pure.
			For $\ket{\phi}$, we have:
			\begin{align}
				&S(B_1)+S(B_2C_1)+S(C_2)=S(B_1B_2C_1C_2)=S(AD)_\phi 
				= S(A)_\phi + S(D)_\phi.
			\end{align}
			However, notice that $S(B_1)=S(A)_\phi$ and $S(C_2)=S(D)_{\phi'}$ due the structure \cref{eq:defphiphi'}, and $S(D)_{\phi'}=S(D)_{\phi}$ by assumption, we are forced to have $S(B_2C_1)=0$.
			
			It follows that: 
			\begin{align}
				\ket{\phi} = \ket{\phi}_{AB_1}\ket{\phi}_{B_2C_1}\ket{\phi}_{C_2D},~~~
				\ket{\phi'} = \ket{\phi'}_{AB_1}\ket{\phi}_{B_2C_1}\ket{\phi'}_{C_2D}.
			\end{align}
			To glue $\ket{\psi}$ and $\ket{\psi'}$, we simply replace $\ket{\phi'}_{AB_1}$ by $\ket{\phi}_{AB_1}$ then map the states back using $V_BV_C$.
			More precisely, for any unitary $U_{A}$ such that $\ket{\phi}_{AB_1}=U_{A}\ket{\phi'}_{AB_1}$, which always exists due to the uniqueness of purification, we define:
			\begin{equation}
				\begin{aligned}
					&\ket{\Phi} = U_A\ket{\phi}=\ket{\phi}_{AB_1}\ket{\phi}_{B_2C_1}\ket{\phi'}_{C_2D},\\
					&\ket{\Psi} = V_BV_C \ket{\Phi}=U_A\ket{\psi'}.    
				\end{aligned}
			\end{equation}
			The desired properties \cref{eq:gluereq} are then clearly satisfied.
		\end{proof}
		
		\textbf{Remark}.
		The unitary $U_A$ can be determined solely from $\phi_{AB_1}$ and $\phi'_{AB_1}$,
		hence solely from $\psi_{AB}$, $\psi'_{AB}$ and $V_B$.
		
		Another canonical choice, although not a unitary, is the Petz map $\mathcal{P}_{B\to AB}$ defined by $\psi$:
		\begin{equation}
			\mathcal{P}_{B\to AB}(\cdot)=\psi_{AB}^{\frac{1}{2}}\psi_{B}^{-\frac{1}{2}}(I_A\otimes\cdot)\psi_{B}^{-\frac{1}{2}}\psi_{AB}^{\frac{1}{2}}.
		\end{equation}
		Explicit calculation shows that $\mathcal{P}_{B\to AB}(\psi_{BCD})=\ketbra{\psi}$, $\mathcal{P}_{B\to AB}(\psi'_{BCD})=\ketbra{\Psi}$.
		
		\section{Double Fibonacci Model}\label{sec:doubleFib}
		\begin{proposition}
			The double Fibonacci model on torus does not admit locality-preserving logical gates.
		\end{proposition}
		The proposition is proved by combing the following two results:
		\begin{itemize}
			\item  Ref.~\cite{beverland2016protected} shows that $ULU^\dagger$ must be monomial for any locality-preserving logical gate $L$ and any modular transformation $U$;
			\item  Ref.~\cite{beverland2016toward} (sec. 3.9) shows that the permutation pattern of $L$ must preserve the quantum dimension. 
		\end{itemize}
		
		\begin{proof}
			For the double Fibonacci model, the modular transformations are generated by:
			\begin{equation}
				S=\frac{1}{2+\phi} \begin{pmatrix} 1 & \phi & \phi & \phi^2 \\ \phi & -1 & \phi^2 & -\phi \\ \phi & \phi^2 & -1 & -\phi \\ \phi^2 & -\phi & -\phi & 1 \end{pmatrix},~~
				T = \text{diag}\left( 1,  e^{i 4\pi/5},  e^{-i 4\pi/5},  1 \right),
			\end{equation}
			where $\phi=\frac{1+\sqrt{5}}{2}$ satisfies $\phi^2=\phi+1$.
			There are 4 anyon types, whose quantum dimensions are $1, \phi, \phi, \phi^2$ respectively,
			
			A locality-preserving logical gate $L$ has the form of $L=\Pi_\pi D$.
			Here, we use $\Pi_\pi$ to denote the permutation matrix corresponding to the permutation $\pi$.
			Namely, $\Pi_{i,j}=1$ iff $\pi(i)=j$, and $\Pi_{i,j}=0$ otherwise.
			To preserve the quantum dimension, only two permutations are possible: $\pi=(1)(2)(3)(4)$ or $\pi=(1)(23)(4)$.
			We parametrize $D$ as $D=\text{diag}(1,z_1,z_2,z_3)$, where $|z_i|=1$.
			
			(1) If $\pi=(1)(2)(3)(4)$.
			Computing $SLS^\dagger$, we get:
			\begin{equation}
				\footnotesize
				\frac{1}{5\phi^2}\begin{pmatrix*}
					1+\phi^{2}z_1+\phi^{2}z_2+\phi^{4}z_3 &
					\phi-\phi z_1+\phi^{3}z_2-\phi^{3}z_3 &
					\phi+\phi^{3}z_1-\phi z_2-\phi^{3}z_3 &
					\phi^{2}-\phi^{2}z_1-\phi^{2}z_2+\phi^{2}z_3 \\
					\phi-\phi z_1+\phi^{3}z_2-\phi^{3}z_3 &
					\phi^{2}+z_1+\phi^{4}z_2+\phi^{2}z_3 &
					\phi^{2}-\phi^{2}z_1-\phi^{2}z_2+\phi^{2}z_3 &
					\phi^{3}+\phi z_1-\phi^{3}z_2-\phi z_3 \\
					\phi+\phi^{3}z_1-\phi z_2-\phi^{3}z_3 &
					\phi^{2}-\phi^{2}z_1-\phi^{2}z_2+\phi^{2}z_3 &
					\phi^{2}+\phi^{4}z_1+z_2+\phi^{2}z_3 &
					\phi^{3}-\phi^{3}z_1+\phi z_2-\phi z_3 \\
					\phi^{2}-\phi^{2}z_1-\phi^{2}z_2+\phi^{2}z_3 &
					\phi^{3}+\phi z_1-\phi^{3}z_2-\phi z_3 &
					\phi^{3}-\phi^{3}z_1+\phi z_2-\phi z_3 &
					\phi^{4}+\phi^{2}z_1+\phi^{2}z_2+z_3
				\end{pmatrix*}.
				\normalsize
			\end{equation}
			Using $|z_i|=1$, we can check that all off-diagonal elements have modulus less than 1.
			For it to be a monomial matrix, it must be diagonal, hence the off-diagonal elements are strictly 0.
			Solving it, we find $z_1=z_2=z_3=1$.

			(2) If $\pi=(1)(2,3)(4)$.
			Computing $SLS^\dagger$ and solving $z_i$ as in (1), we find $z_1=z_2=z_3=1$.
			Hence $L=\Pi_\pi$.
			Now, computing $STL(ST)^\dagger$, we find that it is not monomial, a contradiction.
			
			Therefore, the only possibility is $L=id$.
		\end{proof}

	\end{appendix}
	
	\bibliographystyle{unsrt}
	\bibliography{ref.bib}

@article{chiralpaper,
    author = {Ellison, Tyler and Lee, Dongjin and Li, Zhi and Moharramipour, Amin and Panahi, Yasamin and Yoshida, Beni},
    title = {{Many-body chirality of topological stabilizer states}},
    year = {2026 (to appear)}
}

@article{andreadakis2026exact,
  title={Exact link between nonlocal nonstabilizerness and operator entanglement},
  author={Andreadakis, Faidon and Zanardi, Paolo},
  journal={Physical Review A},
  volume={113},
  number={1},
  pages={L010404},
  year={2026},
  publisher={APS}
}

@article{dur2002effective,
  title={Effective size of certain macroscopic quantum superpositions},
  author={D{\"u}r, Wolfgang and Simon, Christoph and Cirac, J Ignacio},
  journal={Physical review letters},
  volume={89},
  number={21},
  pages={210402},
  year={2002},
  publisher={APS}
}

@article{zhang2024long,
  title={Non-onsite symmetry breaking: Topological phase coexistence and criticality},
  author={Zhang, Zhehao and Li, Yabo and Lu, Tsung-Cheng},
  journal={Physical Review B},
  volume={113},
  number={12},
  pages={125123},
  year={2026},
  publisher={APS}
}

@article{smith2023deterministic,
  title={Deterministic constant-depth preparation of the AKLT state on a quantum processor using fusion measurements},
  author={Smith, Kevin C and Crane, Eleanor and Wiebe, Nathan and Girvin, SM},
  journal={PRX Quantum},
  volume={4},
  number={2},
  pages={020315},
  year={2023},
  publisher={APS}
}

@article{smith2024constant,
  title={Constant-depth preparation of matrix product states with adaptive quantum circuits},
  author={Smith, Kevin C and Khan, Abid and Clark, Bryan K and Girvin, Steven M and Wei, Tzu-Chieh},
  journal={PRX Quantum},
  volume={5},
  number={3},
  pages={030344},
  year={2024},
  publisher={APS}
}

@article{sahay2025classifying,
  title={Classifying one-dimensional quantum states prepared by a single round of measurements},
  author={Sahay, Rahul and Verresen, Ruben},
  journal={PRX Quantum},
  volume={6},
  number={1},
  pages={010329},
  year={2025},
  publisher={APS}
}

@article{stephen2025preparing,
  title={Preparing matrix product states via fusion: Constraints and extensions},
  author={Stephen, David T and Hart, Oliver},
  journal={Physical Review B},
  volume={112},
  number={23},
  pages={235127},
  year={2025},
  publisher={APS}
}

@article{zhang2024characterizing,
  title={Characterizing matrix-product states and projected entangled-pair states preparable via measurement and feedback},
  author={Zhang, Yifan and Gopalakrishnan, Sarang and Styliaris, Georgios},
  journal={PRX Quantum},
  volume={5},
  number={4},
  pages={040304},
  year={2024},
  publisher={APS}
}

@article{chen2010local,
  title={Local unitary transformation, long-range quantum entanglement, wave function renormalization, and topological order},
  author={Chen, Xie and Gu, Zheng-Cheng and Wen, Xiao-Gang},
  journal={Physical Review B—Condensed Matter and Materials Physics},
  volume={82},
  number={15},
  pages={155138},
  year={2010},
  publisher={APS}
}

@article{gottesman1998heisenberg,
  title={The Heisenberg representation of quantum computers},
  author={Gottesman, Daniel},
  journal={arXiv preprint quant-ph/9807006},
  year={1998}
}

@article{bravyi2005universal,
  title={Universal quantum computation with ideal Clifford gates and noisy ancillas},
  author={Bravyi, Sergey and Kitaev, Alexei},
  journal={Physical Review A—Atomic, Molecular, and Optical Physics},
  volume={71},
  number={2},
  pages={022316},
  year={2005},
  publisher={APS}
}

@article{veitch2014resource,
  title={The resource theory of stabilizer quantum computation},
  author={Veitch, Victor and Hamed Mousavian, SA and Gottesman, Daniel and Emerson, Joseph},
  journal={New Journal of Physics},
  volume={16},
  number={1},
  pages={013009},
  year={2014},
  publisher={IOP Publishing}
}

@article{howard2017application,
  title={Application of a resource theory for magic states to fault-tolerant quantum computing},
  author={Howard, Mark and Campbell, Earl},
  journal={Physical review letters},
  volume={118},
  number={9},
  pages={090501},
  year={2017},
  publisher={APS}
}

@article{korbany2025long,
  title={Long-range nonstabilizerness and phases of matter},
  author={Korbany, David Aram and Gullans, Michael J and Piroli, Lorenzo},
  journal={Physical Review Letters},
  volume={135},
  number={16},
  pages={160404},
  year={2025},
  publisher={APS}
}

@article{ellison2021symmetry,
  title={Symmetry-protected sign problem and magic in quantum phases of matter},
  author={Ellison, Tyler D and Kato, Kohtaro and Liu, Zi-Wen and Hsieh, Timothy H},
  journal={Quantum},
  volume={5},
  pages={612},
  year={2021},
  publisher={Verein zur F{\"o}rderung des Open Access Publizierens in den Quantenwissenschaften}
}

@article{white2021conformal,
  title={Conformal field theories are magical},
  author={White, Christopher David and Cao, ChunJun and Swingle, Brian},
  journal={Physical Review B},
  volume={103},
  number={7},
  pages={075145},
  year={2021},
  publisher={APS}
}

@article{bravyi2010tradeoffs,
  title={Tradeoffs for reliable quantum information storage in 2D systems},
  author={Bravyi, Sergey and Poulin, David and Terhal, Barbara},
  journal={Physical review letters},
  volume={104},
  number={5},
  pages={050503},
  year={2010},
  publisher={APS}
}

@article{bravyi2025much,
  title={How much entanglement is needed for quantum error correction?},
  author={Bravyi, Sergey and Lee, Dongjin and Li, Zhi and Yoshida, Beni},
  journal={Physical Review Letters},
  volume={134},
  number={21},
  pages={210602},
  year={2025},
  publisher={APS}
}

@article{bravyi2006lieb,
  title={Lieb-Robinson bounds and the generation of correlations and topological quantum order},
  author={Bravyi, Sergey and Hastings, Matthew B and Verstraete, Frank},
  journal={Physical review letters},
  volume={97},
  number={5},
  pages={050401},
  year={2006},
  publisher={APS}
}

@article{li2025much,
  title={How much entanglement is needed for topological codes and mixed states with anomalous symmetry?},
  author={Li, Zhi and Lee, Dongjin and Yoshida, Beni},
  journal={Physical Review X},
  volume={15},
  number={2},
  pages={021090},
  year={2025},
  publisher={APS}
}

@article{haah2016invariant,
  title={An invariant of topologically ordered states under local unitary transformations},
  author={Haah, Jeongwan},
  journal={Communications in Mathematical Physics},
  volume={342},
  number={3},
  pages={771--801},
  year={2016},
  publisher={Springer}
}

@article{aharonov2018quantum,
  title={Quantum circuit depth lower bounds for homological codes},
  author={Aharonov, Dorit and Touati, Yonathan},
  journal={arXiv preprint arXiv:1810.03912},
  year={2018}
}

@article{lu2022measurement,
  title={Measurement as a shortcut to long-range entangled quantum matter},
  author={Lu, Tsung-Cheng and Lessa, Leonardo A and Kim, Isaac H and Hsieh, Timothy H},
  journal={PRX Quantum},
  volume={3},
  number={4},
  pages={040337},
  year={2022},
  publisher={APS}
}

@article{verresen2021efficiently,
  title={Efficiently preparing Schrödinger's cat, fractons and non-Abelian topological order in quantum devices},
  author={Verresen, Ruben and Tantivasadakarn, Nathanan and Vishwanath, Ashvin},
  journal={arXiv preprint arXiv:2112.03061},
  year={2021}
}

@article{tantivasadakarn2024long,
  title={Long-range entanglement from measuring symmetry-protected topological phases},
  author={Tantivasadakarn, Nathanan and Thorngren, Ryan and Vishwanath, Ashvin and Verresen, Ruben},
  journal={Physical Review X},
  volume={14},
  number={2},
  pages={021040},
  year={2024},
  publisher={APS}
}

@article{tantivasadakarn2023hierarchy,
  title={Hierarchy of topological order from finite-depth unitaries, measurement, and feedforward},
  author={Tantivasadakarn, Nathanan and Vishwanath, Ashvin and Verresen, Ruben},
  journal={PRX Quantum},
  volume={4},
  number={2},
  pages={020339},
  year={2023},
  publisher={APS}
}

@article{kim2025any,
  title={Any Clifford+ T circuit can be controlled with constant T-depth overhead},
  author={Kim, Isaac H and Laakkonen, Tuomas},
  journal={arXiv preprint arXiv:2512.24982},
  year={2025}
}

@article{bravyi2022adaptive,
  title={Adaptive constant-depth circuits for manipulating non-abelian anyons},
  author={Bravyi, Sergey and Kim, Isaac and Kliesch, Alexander and Koenig, Robert},
  journal={arXiv preprint arXiv:2205.01933},
  year={2022}
}

@article{kitaev2003fault,
  title={Fault-tolerant quantum computation by anyons},
  author={Kitaev, A Yu},
  journal={Annals of physics},
  volume={303},
  number={1},
  pages={2--30},
  year={2003},
  publisher={Elsevier}
}

@article{levin2005string,
  title={String-net condensation: A physical mechanism for topological phases},
  author={Levin, Michael A and Wen, Xiao-Gang},
  journal={Physical Review B—Condensed Matter and Materials Physics},
  volume={71},
  number={4},
  pages={045110},
  year={2005},
  publisher={APS}
}

@article{bravyi2013classification,
  title={Classification of topologically protected gates for local stabilizer codes},
  author={Bravyi, Sergey and K{\"o}nig, Robert},
  journal={Physical review letters},
  volume={110},
  number={17},
  pages={170503},
  year={2013},
  publisher={APS}
}

@article{eastin2009restrictions,
  title={Restrictions on transversal encoded quantum gate sets},
  author={Eastin, Bryan and Knill, Emanuel},
  journal={Physical review letters},
  volume={102},
  number={11},
  pages={110502},
  year={2009},
  publisher={APS}
}

@article{verlinde1988fusion,
  title={Fusion rules and modular transformations in 2D conformal field theory},
  author={Verlinde, Erik},
  journal={Nuclear Physics B},
  volume={300},
  pages={360--376},
  year={1988},
  publisher={Elsevier}
}

@article{yi2025lov,
  title={Lovász Meets {Lieb-Schultz-Mattis}: Complexity in Approximate Quantum Error Correction},
  author={Yi, Jinmin and Liu, Ruizhi and Li, Zhi},
  journal={arXiv preprint arXiv:2510.04453},
  year={2025}
}

@article{nikshych2014categorical,
  title={Categorical Lagrangian Grassmannians and Brauer--Picard groups of pointed fusion categories},
  author={Nikshych, Dmitri and Riepel, Brianna},
  journal={Journal of Algebra},
  volume={411},
  pages={191--214},
  year={2014},
  publisher={Elsevier}
}

@article{anshu2020circuit,
  title={Circuit lower bounds for low-energy states of quantum code Hamiltonians},
  author={Anshu, Anurag and Nirkhe, Chinmay},
  journal={arXiv preprint arXiv:2011.02044},
  year={2020}
}

@article{wei2025long,
  title={Long-range nonstabilizerness from quantum codes, orders, and correlations},
  author={Wei, Fuchuan and Liu, Zi-Wen},
  journal={arXiv preprint arXiv:2503.04566},
  year={2025}
}

@article{muger2003subfactors,
  title={From subfactors to categories and topology II: The quantum double of tensor categories and subfactors},
  author={M{\"u}ger, Michael},
  journal={Journal of Pure and Applied Algebra},
  volume={180},
  number={1-2},
  pages={159--219},
  year={2003},
  publisher={Elsevier}
}

@phdthesis{beverland2016toward,
  title={Toward realizable quantum computers},
  author={Beverland, Michael E},
  year={2016},
  school={California Institute of Technology}
}

@article{freedman2001two,
  title={The two-eigenvalue problem and density of Jones representation of braid groups},
  author={Freedman, Michael H and Larsen, Michael J and Wang, Zhenghan},
  journal={arXiv preprint math/0103200},
  year={2001}
}

@article{larsen2005density,
  title={Density of the {SO(3) TQFT} representation of mapping class groups},
  author={Larsen, Michael and Wang, Zhenghan},
  journal={Communications in mathematical physics},
  volume={260},
  number={3},
  pages={641--658},
  year={2005},
  publisher={Springer}
}

@article{yoshida2015topological,
  title={Topological color code and symmetry-protected topological phases},
  author={Yoshida, Beni},
  journal={Physical Review B},
  volume={91},
  number={24},
  pages={245131},
  year={2015},
  publisher={APS}
}

@article{yoshida2017gapped,
  title={Gapped boundaries, group cohomology and fault-tolerant logical gates},
  author={Yoshida, Beni},
  journal={Annals of Physics},
  volume={377},
  pages={387--413},
  year={2017},
  publisher={Elsevier}
}

@article{ng2010congruence,
  title={Congruence subgroups and generalized Frobenius-Schur indicators},
  author={Ng, Siu-Hung and Schauenburg, Peter},
  journal={Communications in Mathematical Physics},
  volume={300},
  number={1},
  pages={1--46},
  year={2010},
  publisher={Springer}
}

@article{cui2020kitaev,
  title={Kitaev's quantum double model as an error correcting code},
  author={Cui, Shawn X and Ding, Dawei and Han, Xizhi and Penington, Geoffrey and Ranard, Daniel and Rayhaun, Brandon C and Shangnan, Zhou},
  journal={Quantum},
  volume={4},
  pages={331},
  year={2020},
  publisher={Verein zur F{\"o}rderung des Open Access Publizierens in den Quantenwissenschaften}
}

@article{qiu2020ground,
  title={Ground subspaces of topological phases of matter as error correcting codes},
  author={Qiu, Yang and Wang, Zhenghan},
  journal={Annals of Physics},
  volume={422},
  pages={168318},
  year={2020},
  publisher={Elsevier}
}

@article{westwick1967transformations,
  title={Transformations on tensor spaces},
  author={Westwick, Roy},
  journal={Pacific Journal of Mathematics},
  volume={23},
  number={3},
  pages={613--620},
  year={1967},
  publisher={Mathematical Sciences Publishers}
}

@article{parham2025quantum,
  title={Quantum circuit lower bounds in the magic hierarchy},
  author={Parham, Natalie},
  journal={arXiv preprint arXiv:2504.19966},
  year={2025}
}

@article{yi2024complexity,
  title={Complexity and order in approximate quantum error-correcting codes},
  author={Yi, Jinmin and Ye, Weicheng and Gottesman, Daniel and Liu, Zi-Wen},
  journal={Nature Physics},
  volume={20},
  number={11},
  pages={1798--1803},
  year={2024},
  publisher={Nature Publishing Group UK London}
}

@article{beverland2016protected,
  title={Protected gates for topological quantum field theories},
  author={Beverland, Michael E and Buerschaper, Oliver and Koenig, Robert and Pastawski, Fernando and Preskill, John and Sijher, Sumit},
  journal={Journal of Mathematical Physics},
  volume={57},
  number={2},
  year={2016},
  publisher={AIP Publishing}
}

@article{arad2013area,
  title={An area law and sub-exponential algorithm for 1D systems},
  author={Arad, Itai and Kitaev, Alexei and Landau, Zeph and Vazirani, Umesh},
  journal={arXiv preprint arXiv:1301.1162},
  year={2013}
}

@article{kato2016information,
  title={Information-theoretical analysis of topological entanglement entropy and multipartite correlations},
  author={Kato, Kohtaro and Furrer, Fabian and Murao, Mio},
  journal={Physical Review A},
  volume={93},
  number={2},
  pages={022317},
  year={2016},
  publisher={APS}
}
\end{document}